\documentclass[aps,prl,a4paper,superscriptaddress,twocolumn]{revtex4-1}

\usepackage[utf8,latin1]{inputenc}
\usepackage[T1]{fontenc} 
\usepackage[ngerman,greek,brazilian,british,american]{babel}
\usepackage{lipsum}
\usepackage{float}
\usepackage[dvipsnames]{xcolor}
\usepackage[colorlinks=true,citecolor=Aquamarine,linkcolor=Mulberry,urlcolor=Magenta,hyperindex]{hyperref}
\usepackage{graphicx}
\usepackage{epsfig}
\usepackage{color}
\usepackage{lmodern}
\usepackage{mathpazo}
\usepackage[babel=true]{microtype}
\usepackage{amsmath,amssymb,amsthm}
\usepackage{mathtools}
\usepackage{bm}
\usepackage{cleveref}
\usepackage{enumerate}
\usepackage{sidecap}
\usepackage{enumitem}
\usepackage{listings}
\usepackage{IEEEtrantools}

\newcommand{\ket}[1]{\vert#1\rangle}

\newcommand{\ketbra}[2]{\vert #1 \rangle \langle #2 \vert}
\newcommand{\avg}[1]{\langle#1\rangle}

\newcommand{\comm}[2]{[#1, #2]}
\newcommand{\bcomm}[2]{\bigl[#1, #2\bigr]}
\newcommand{\id}{\openone}

\DeclareMathOperator{\conv}{Conv}

\newcommand{\bigro}[1]{\bigl(#1\bigr)}
\newcommand{\Bigro}[1]{\Bigl(#1\Bigr)}

\makeatletter
\newcommand*{\Rnum}[1]{\expandafter\@slowromancap\romannumeral #1@}
\makeatother

\newcommand{\ba}{\begin{eqnarray}}
\newcommand{\be}{\begin{equation}}
\newcommand{\ee}{\end{equation}}
\newcommand{\beq}{\begin{equation}}
\newcommand{\eeq}{ \end{equation}}
\newcommand{\bea}{\begin{eqnarray}}
\newcommand{\eea}{ \end{eqnarray}}

\newcommand{\mean}[1]{\left\langle#1\right\rangle}
\DeclarePairedDelimiter{\abs}{\lvert}{\rvert}
\newcommand{\babs}[1]{\abs[big]{#1}}

\newcommand{\ea}{\end{eqnarray}}
\newcommand{\ban}{\begin{eqnarray*}}
\newcommand{\ean}{\end{eqnarray*}}

\newcommand{\tr}{\text{\normalfont tr}}

\newcommand{\etal}{\textit{et al}. }
\newcommand{\ie}{\textit{i.e.}}
\newcommand{\eg}{\textit{e.g.}}

\newcommand{\red}[1]{\textcolor{red}{#1}}

\newcommand{\blue}[1]{\textcolor{blue}{#1}}

\newtheorem{definition}{Definition}
\newtheorem{theorem}{Theorem}

\begin{document}

\title{Device-independent tests of structures of measurement incompatibility}

	
\author{Marco T\'ulio Quintino}
\affiliation{Department of Physics, Graduate School of Science, The University of Tokyo, Hongo 7-3-1, Bunkyo-ku, Tokyo 113-0033, Japan}

\author{Costantino Budroni}
\affiliation{Institute for Quantum Optics and Quantum Information (IQOQI),
Austrian Academy of Sciences, Boltzmanngasse 3, 1090 Vienna, Austria}
\affiliation{Faculty of Physics, University of Vienna, Boltzmanngasse 5, 1090 Vienna, Austria}

\author{Erik Woodhead}
\affiliation{ICFO-Institut de Ci\`encies Fotoniques, The Barcelona Institute of Science and Technology, 08860 Castelldefels (Barcelona), Spain}

\author{Ad\'an Cabello}
\affiliation{Departamento de F\'{\i}sica Aplicada II, Universidad de Sevilla, E-41012 Sevilla, Spain}
\affiliation{Instituto Carlos~I de F\'{\i}sica Te\'orica y Computacional, Universidad de Sevilla, E-41012 Sevilla, Spain}

\author{Daniel Cavalcanti}
\affiliation{ICFO-Institut de Ci\`encies Fotoniques, The Barcelona Institute of Science and Technology, 08860 Castelldefels (Barcelona), Spain}

\date{\today}


\begin{abstract}
In contrast with classical physics, in quantum physics some sets of measurements are incompatible in the sense that they can not be performed simultaneously. Among other applications, incompatibility allows for contextuality and Bell nonlocality. This makes {it} of crucial importance {to develop} tools for certifying whether a set of measurements respects a certain structure of incompatibility. Here we show that, for quantum or nonsignaling models, if the measurements employed in a Bell test satisfy a given type of compatibility, then the amount of violation of some specific Bell inequalities {becomes} limited. Then, we show that correlations arising from local measurements on two-qubit states violate these limits, which rules out in a device-independent way such structures of incompatibility. In particular, we prove that quantum correlations allow for a device-independent demonstration of genuine triplewise incompatibility. Finally, we translate these results into a semi-device-independent Einstein-Podolsky-Rosen-steering scenario. 
\end{abstract}


\maketitle

The fact that some pairs of quantum observables do not commute implies that they can not be measured simultaneously as the corresponding operators do not share a common set of eigenvectors \cite{krausbook}. This incompatibility property of quantum measurements is used in several quantum information protocols such as quantum cryptography \cite{cripto_review} and quantum state discrimination \cite{CHT18,UKSYG18,SSC19}, and is also required in proofs of contextuality \cite{amaral2018graph,XC18}, Einstein-Podolsky-Rosen steering (EPR-steering) \cite{quintino14,uola14}, and Bell nonlocality \cite{NL_review}. 

It is thus of fundamental and practical importance to develop tools to experimentally certify that a set of measurements respects a given type of incompatibility, required for producing a specific type of quantum {correlation}. Moreover, it would be very useful to be able to achieve such a certification without needing to model the experimental procedures that generate the experimental statistics. This is precisely the aim of the paradigm of device-independent certification used, for instance, for certifying secure communication \cite{acin07} and randomness \cite{randomness_review}. This paradigm assumes that {quantum theory (QT)} is correct and that {signaling} between {spacelike} separated events is impossible. Then, it uses the violation of specifically tailored Bell inequalities \cite{bell64} to certify a targeted property using only the experimental statistics.

The relation between Bell inequality violation and measurement incompatibility was first studied by Fine, who showed that, in the scenario where two parties are restricted to dichotomic measurements, a Bell inequality can only be violated if the observers use incompatible measurements \cite{fine82}. Later, Wolf \etal \cite{wolf09} showed that every pair of incompatible measurement can be used to violate the simplest Bell inequality, namely the Clauser-Horne-Shimony-Holt (CHSH) inequality \cite{chsh69}. Moreover, methods for device-independent quantification of incompatibility have been proposed \cite{cavalcantiPRA16,chenPRL16,chenPRA18} and it is known that some {sets} of incompatible measurements can not be used to violate Bell inequalities \cite{quintino15b,HQB18,bene17}. Finally, it is known that when more than two measurements are considered, different compatibility structures may appear \cite{teiko08,liang11}. 


\begin{figure}
	\includegraphics[scale=.43]{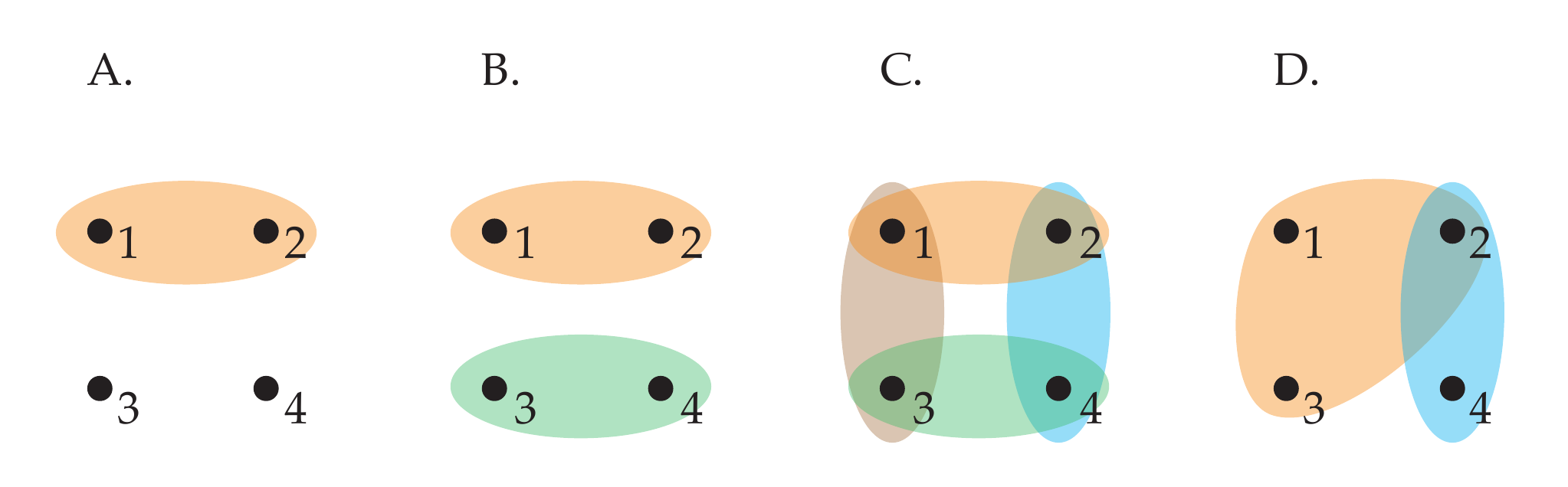} 
	\caption{Examples of incompatibility structures for four measurements. Each node represents a complete measurement (\ie{} a complete set of POVM elements), and each hyper-edge (represented by a region colored with the same color) contains measurements that are compatible. If a set of measurements are not contained in a hyper-edge they are incompatible. The respective incompatibility structures are represented by the following {hypergraph}s: A.: $\mathcal{C}_{\text{A}}=[(1,2)]$; B.: $\mathcal{C}_{\text{B}}=[(1,2),(3,4)]$; C.: $\mathcal{C}_{\text{C}}=[(1,2),(1,3),(2,4),(3,4)]$; D.: $\mathcal{C}_{\text{D}}=[(1,2,3),(2,4)]$.}\label{fig:compatibility}
\end{figure}


In this Letter, we show how to test if a specific structure of incompatibility {is required to} generate the statistics observed in a Bell test. Our approach is based on the intuition that, if the measurements used in the Bell test satisfy a targeted structure of {compatibility}, then the amount of Bell violation becomes limited and, therefore, any violation beyond this limit rules out the presence of the targeted {compatibility} structure. We also show examples of such violations in the simplest scenario of local measurements applied to two-qubit systems. Thus, at least the simplest structures of incompatibility can be certified in a device-independent way.


{\em Pairwise and $n$-wise incompatibility.---}In QT, measurements on $d$-dimensional quantum systems are described by positive operators $M_{a|x}\geq0$ (we use $x$ to label different measurements and $a$ their outcomes) acting on a $d$-dimensional complex Hilbert space $\mathbb{C}^d$ and satisfying the normalization condition $\sum_a M_{a|x} = \openone, \quad \forall x$ ($\openone$ is the identity operator). A set of quantum measurements is compatible if and only if there exists a set of measurement operators $\{E_\lambda\}$ ($E_\lambda\geq0$ and $\sum_\lambda E_\lambda=\openone$) such that 
\ba\label{compatible M}
M_{a|x}=\sum_\lambda p(a|x,\lambda) E_\lambda,~\forall a,x,
\ea
where $p(a|x,\lambda)\geq0$ and $\sum_a p(a|x,\lambda)=1~ \forall x, \lambda$ \cite{kru}. Otherwise, they are incompatible. Notice that a set of compatible measurements can be implemented simultaneously by employing the measurement $\{E_\lambda\}$ and post-processing the results according to the probability {distribution} $\{p(a|x,\lambda)\}$. 

Given the previous definition, a set of measurements can present different structures of compatibility. For instance, a set of three measurements can be pairwise compatible but incompatible when all three measurements are considered \cite{teiko08}. In general the compatibility structure of a set of measurements can be represented by a hyper-graph $\mathcal{C}=[{C}_1,{C}_2,...,{C}_k]$, where each hyper-edge ${C}_i$ indicates a subset of measurements that are compatible. For instance, the structure $\mathcal{C}_{\text{pair}}=[\{1,2\},\{1,3\},\{2,3\}]$ indicates that the measurements $1,2$ and $3$ are pairwise compatible, but not {triplewise} compatible, while the structure $\mathcal{C}_{\text{3full}}=[\{1,2,3\}]$ indicates full {triplewise} compatibility (see Fig.~\ref{fig:compatibility} for more examples). In the Appendix we show how the different kinds of measurement incompatibility can be tested by {semidefinite} programming. 

Within this framework, we can also define genuine triplewise (or in general $n$-wise) incompatibility: A set of three measurements is genuinely triplewise incompatible when it cannot be written as {a convex combination} of measurements that are pairwise compatible on {different partitions}. Let us illustrate this concept with an example. Consider a set of three noisy qubit Pauli measurements given by measurement operators
\begin{equation} \label{WN}
	M_{a|x}^\eta:= \eta \Pi_{a|x} + (1-\eta) \frac{\openone}{2},
\end{equation}
where $x=1,2,3$ {refers} to each Pauli measurement ($X, Y, Z$), respectively, and $\Pi_{a|x}$ are their eigenprojectors. These measurements are triplewise compatible for $\eta\leq 1/\sqrt{3}\approx 0.58$ and pairwise compatible for $\eta\leq 1/\sqrt{2}\approx 0.71$ \cite{teiko08}. It turns out that, for $\eta\leq \frac{\sqrt{2}+1}{3} \approx 0.80$, the set can be written as a convex combination of other sets in which two measurements are compatible (see Fig.~\ref{fig:genuine}). Thus, although for $\eta>1/\sqrt{2}$ the measurements are triplewise incompatible, it is only for $\eta>\frac{\sqrt{2}+1}{3}$ that they are genuinely triplewise incompatible.


\begin{figure}
\includegraphics[scale=.43]{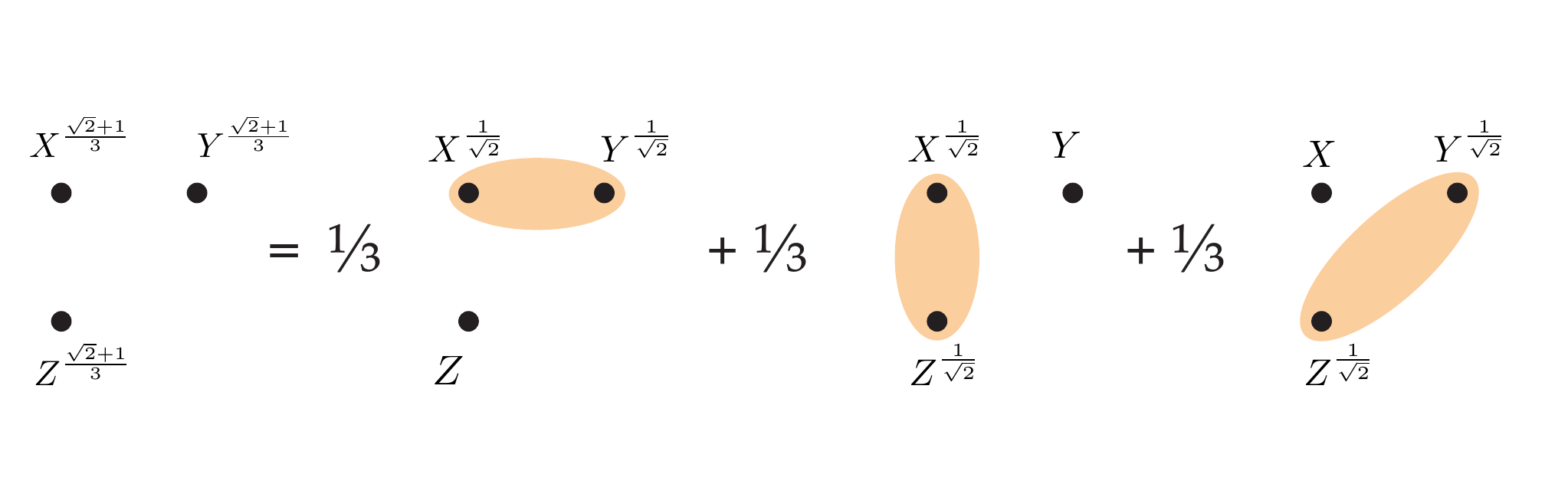}
\caption{The set of noisy Pauli measurements $X^\eta$, $Y^\eta$, $Z^\eta$ defined by  \eqref{WN} for $\eta=\frac{\sqrt{2}+1}{3}$ can be written as {a} uniform convex combination of Pauli measurements that have a compatible pair (represented in a shaded area). Thus, one can implement these measurements by randomly implementing sets of measurements which are not triplewise incompatible.}\label{fig:genuine}
\end{figure}


{\em Device-independent test of structures of incompatibility.---}We now turn to the question of certifying the different types of measurement incompatibility in a device-independent way, \ie{}, by analyzing the statistics of input and outputs relations of measurements. We consider a bipartite Bell scenario where two parties, Alice and Bob, share a bipartite state $\rho$ onto which they perform measurements labeled by $x$ and $y$ with outcomes $a$ and $b$, respectively. After many rounds of the experiment, Alice and Bob can determine the set of conditional probability distributions $\{p(ab|xy)\}$, which we call the observed \emph{behavior} \cite{Tsirelson93}. A behavior is local when it can be written as \cite{NL_review}
\begin{equation}
  \label{bell_local}
  p(ab|xy)=\sum_\lambda p(\lambda) p_A(a|x,\lambda)p_B(b|y,\lambda), \forall a,b,x,y,
\end{equation}
where $p(\lambda)$, $p_A(a|x,\lambda)$, and $p(b|y,\lambda)$ are probability distributions. We denote the set of local behaviors {by} $L$.

If one of the parties, say Alice, performs a set of measurements which are fully compatible, the observed behavior is local regardless the shared state and the measurements of Bob \cite{fine82}. 
This can be explicitly seen by using the definition \eqref{compatible M} as follows:
\ba \label{JM_NL}
p(ab|xy)&=&\tr(M_{a|x}\otimes M_{b|y} \rho)\\ \
&=&\sum_\lambda p_A(a|x,\lambda) \tr(E_\lambda\otimes M_{b|y} \rho)\nonumber\\
&=&\sum_\lambda p_A(a|x,\lambda)p_B(b,\lambda|y)\nonumber\\
&=&\sum_\lambda p(\lambda) p_A(a|x,\lambda)p_B(b|y,\lambda).\nonumber
\ea
It then follows that the observation of a nonlocal behavior (or equivalently the violation of a Bell inequality) certifies in a device-independent way that both parties used incompatible measurements. 

Similarly, in the case that Alice performs a set of measurement that satisfy a more general compatibility structure $\mathcal{C}$, the observed behavior will be local when restricted to the measurements in the hyper-edges $C_i$ of~$\mathcal{C}$. For instance, let us consider the case of three measurements {on} Alice's side for the sake of simplicity. Let $A$ represent a condition {that} the collected total {behavior} is guaranteed to satisfy, for instance, $A$ can be the {nonsignaling} condition ($NS$) or {that} it has a quantum realization in terms of local measurements on a quantum state ($Q$). For any behavior respecting the condition $A$, if Alice's measurements $x=1$ and $x=2$ are compatible, the probabilities of this behavior respect $p(ab|xy)=\sum_\lambda p(\lambda) p_A(a|x,\lambda)p_B(b|y,\lambda)$ for $ x=1,2 $ and any $y$. We {denote $L_{12}^A$ this} set of {behaviors, that respects} the condition $A$ {and is} Bell local for $x=1$ {and} $x=2$.

Notice that the observation that $\{p(ab|xy)\} \notin L_{12}^A$ allows us to conclude that the measurements~1 and 2 are incompatible. Analogously, we can define the sets $L_{23}^A$ and $L_{13}^A$ that correspond to the case where the other pairs of Alice's measurements are compatible. With these three sets representing pairwise compatibility, we can also define their convex hull $L_{2\text{conv}}^A:= \text{Conv}\left( L_{12}^A,L_{23}^A,L_{13}^A\right)$ and intersection $L_{2\cap }^A:= L_{12}^A\cap L_{23}^A\cap L_{13}^A$ (see Fig.~\ref{fig:geometrical2}). The observation that a behavior does not belong to these sets allows us to conclude:

\begin{itemize}
\item If $\{p(ab|xy)\}\notin L$, then there is some incompatibility in Alice's measurements.
\item If $\{p(ab|xy)\}\notin L_{12}^A$, then the measurements $x=1$ and $x=2$ are incompatible.
\item If $\{p(ab|xy)\}\notin L_{2\text{conv}}^A$, then the measurements of Alice are genuinely triplewise incompatible.
\item If $\{p(ab|xy)\}\notin L_{2\cap }^A$, then there is some pairwise incompatibility on Alice's measurements.
\end{itemize}
Notice that we can also define similar sets with respect to Bob's measurements and {consider sets} generated by {given compatibility structures} in {Alice's measurements} and {others} in {Bob's}. 

In what follows, we show that using a set of measurements that satisfy a compatibility structure bounds the amount of violation of certain Bell inequalities. Thus, the observation of a value higher than this bound serves as a certificate that the measurements are incompatible with respect to to this structure. To find these bounds, we need to solve the following optimization problem: given a Bell expression $S=\sum_{abxy} c_{abxy} p(ab|xy)$ and a compatibility structure $\mathcal{C}$,
\ba \label{e:max ineq}
\text{maximize}&\quad S \\
\text{such that}&\quad p(ab|xy) \in L_\mathcal{C} \nonumber\\
&\quad p(ab|xy) \in Q,\nonumber
\ea
where $L_\mathcal{C}$ indicates the set of behaviors {that} are partially local according to the compatibility structure $\mathcal{C}$. Geometrically, this problem can be seen as a maximization of $S$ w.r.t.\ to a set of behaviors that are quantum and satisfy some partial locality (such as the sets $L^Q_{ij}$ in Fig.~\ref{fig:geometrical2}A). The last constraint in \eqref{e:max ineq} imposes that the behavior is quantum (Q), \ie, that it has a quantum realization in terms of local measurements on a quantum state.  In practice, since there is no tractable way of imposing that, we consider sets $Q_n \supseteq Q$ that outer approximate $Q$, {${Q_n}$ being} the $n$-level of the Navascu\'es-Pironio-Ac\'{\i}n (NPA) hierarchy \cite{NPA}. {At} each level $n$, {the} problem {is} a semidefinite program {whose} solution provides an upper bound to the desired bound {and, hence, is} still a valid bound {for detecting} incompatibility.

We emphasize that if Alice performs quantum measurements which are not {genuinely} triplewise incompatible, the {resulting behavior} is inside $L_{2\text{conv}}^Q${; hence} the set $L_{2\text{conv}}^Q$ can be used for device-independent quantum genuine triplewise incompatibility certification.  But since Bell locality does not {necessarily} imply measurement compatibility in general, {the} set $L_{2\text{conv}}^Q$ {may be} larger than the set of quantum {behaviors} generated by imposing that Alice's measurements are not {genuinely} triplewise incompatible. We discuss this in the Appendix where we show that the set of measurements generated by {nongenuinely} triplewise incompatible measurements is strictly smaller than $L_{2\text{conv}}^Q$.



\begin{figure}
	\includegraphics[scale=0.42]{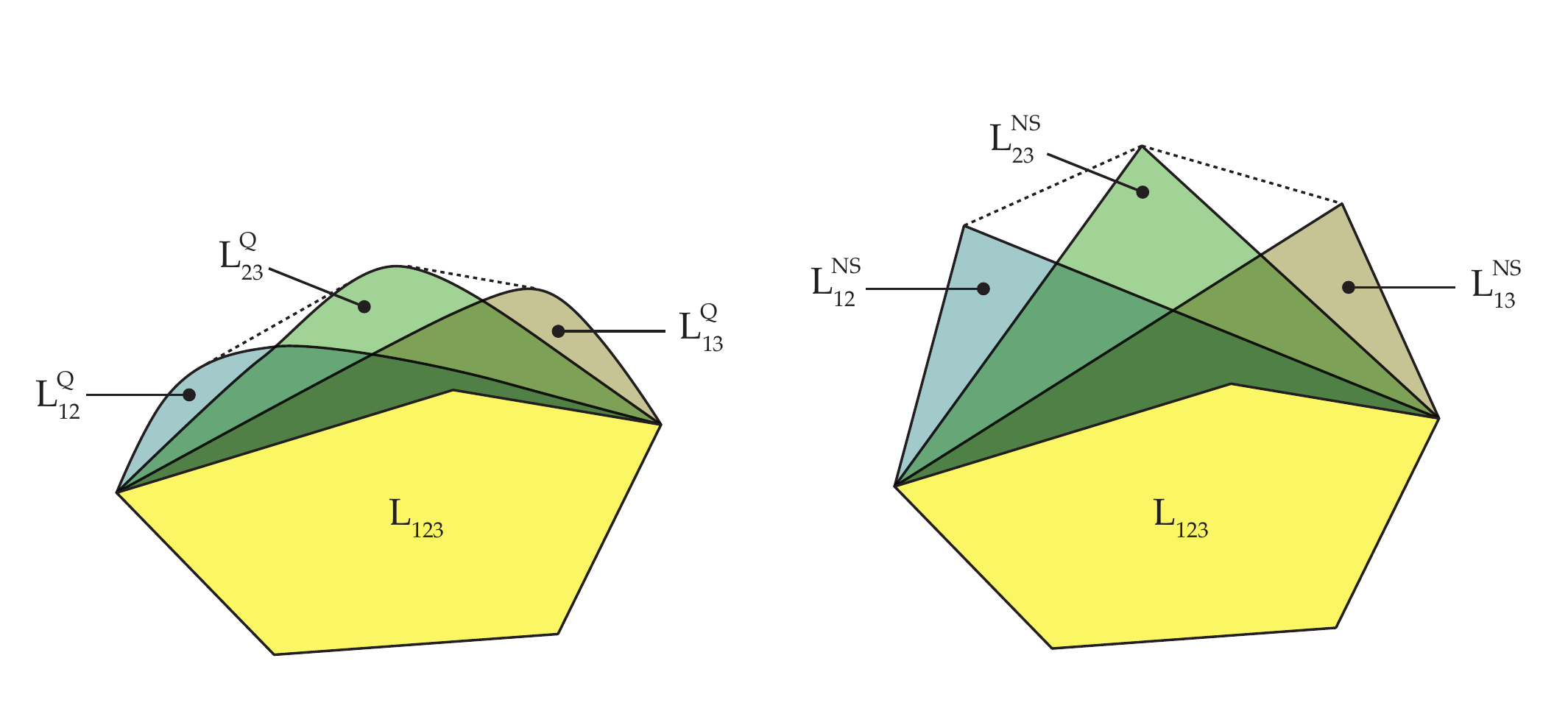} 
	\caption{Geometrical interpretation of sets of three pairwise and triplewise compatible measurements.
		Here $L_{123}$ is the standard local set, where all local measurements are compatible. The set $L_{ij}^{NS}$ consists {of} probabilities that are nonsignaling and {are} partially local w.r.t.\ $i$ and $j$, \ie{}, {are} local when {only} the measurements $i$ and $j$ {are considered on Alice's side}. Analogously, $L_{ij}^{Q}$ is a set of behaviors that are quantum and partially local w.r.t.\ measurements $i$ and $j$.}\label{fig:geometrical2}
\end{figure}


{\em Nonsignaling device-independent witnesses of incompatibility structures.---}It is also possible to test structures of measurement incompatibility not only in QT but in more general nonsignaling theories. For that, we just need to do a similar optimization, but now considering the set of nonsignaling behaviors rather than the set of quantum behaviors. This entails changing the last constraint in~\eqref{e:max ineq} {to} the set of linear constraints {that} defines the general nonsignaling set ${NS}$, \ie, the optimization problem is now
\ba \label{e:max ineq2}
\text{maximize}&\quad S \\
\text{such that}&\quad p(ab|xy) \in L_\mathcal{C} \nonumber\\
&\quad p(ab|xy) \in NS,\nonumber
\ea
 where the last constraint means that the behavior satisfies the nonsignaling conditions
 \begin{subequations}
 \begin{align}
 \sum_a p(ab|x' y)&=\sum_a p(ab|x'' y) ~\forall~x',x'',\\
 \sum_a p(ab|xy')&=\sum_a p(ab|xy'') ~\forall~y',y''. 
 \end{align}
\end{subequations}
Geometrically, this means that the maximization is now running over a bigger set, since $NS\supseteq Q$ (see \eg{} Fig.~\ref{fig:geometrical2}B).

Notice that some of the sets in the problem~\eqref{e:max ineq2}, {which we} denote $L_\mathcal{C}^{\rm NS}$, are easily characterized. In fact, {in the case of dichotomic measurements,} it can be straightforwardly shown that the set $L_{ij}^{\rm NS}$ {is precisely characterized by} the NS constraints plus the {all the} CHSH inequalities involving $A_i,A_j$ and {any} two measurements {on Bob's side}, independently of the number of observables {Bob has}. Similarly, the set $L_{2\cap }^{{\rm NS}}$, obtained as the intersection of the sets for all $ij$, i.e., the union of the systems of inequalities, is described by the NS constraints and all CHSH-type inequalities between {Alice's} and {Bob's} observables.


{\em Results.---}We have run the above optimization problems for a variety of known bipartite Bell expressions $S$ {in scenarios} where Alice has three choices of dichotomic measurements and Bob has three, four, or five choices of dichotomic measurements. After, we completely characterize the polytope $L_{2{\text{conv}}}^{NS}$ by {explicitly} obtaining all the inequalities representing its facets; with that one can easily decide when device-independent certification of genuine triplewise incompatibility is possible if both parties have three dichotomic observables. In order to help {compare} the values, we have set the local bounds of {the} {Bell expressions} to zero and renormalized them such that {their} maximal nonsignaling {bounds are} one. The results are given in Table~\ref{t:results}. 

We first considered all tight Bell inequalities of these scenarios \cite{faacets,quintino14}. Using these inequalities we can test all possible incompatibility structures, including genuine triplewise incompatibility. We then looked at the chained Bell inequality with three inputs \cite{Pearle70,braunstein90}, which is not tight but can be generalized to multiple inputs. We also analyzed the elegant Bell inequality $I_{\text{E}}$ \cite{Elegant} and the chained version of the CHSH inequality proposed in Ref.~\cite{APVW16}, which self-test orthogonal Pauli measurements {on} Alice's side. Although we find quantum violations for every incompatibility structure bound, we did not manage to find a quantum violation of the genuine triplewise incompatibility bounds for general nonsignaling theories. 

In the case of three dichotomic measurements per party we were able to {characterize} the polytope $L_{2\text{conv}}^{NS}$ by {explicitly} obtaining all the inequalities representing its facets (see {the} Appendix). {Among all the} inequalities found there is a single class of inequalities which can be violated by quantum systems, and this inequality is equivalent to the inequality $M_{3322}$ of Ref.~\cite{BGS15}. The $M_{3322}$ inequality can be violated by two-qubit systems and this violation proves that there exist quantum correlations that can not be simulated by any {nonsignaling} model respecting pairwise compatibility. Interestingly, an experimental violation of this inequality was reported in Ref.~\cite{Christensen15}, but the observation of apparent signaling may require a re-analysis of its {conclusions \cite{Smania18,Liang18}.
}

All Bell inequalities tested are explicitly written in the Appendix and the {code} we used {is} available at \cite{mtq_github_incompatibility}.


\begin{table*} 
\begin{tabular}{ccccccccccccc}
\hline \hline 
Ineq. & \;L\; & Qubits & $Q_3$ & $L^{Q_3}_{\cap}$ & $\min{L_{ij}^{Q_3}}$ & $L^{Q_3}_{2\text{conv}}$ & $L^{NS}_{\cap} $ & $\min{L_{ij}^{NS}}$ & $L^{NS}_{2\text{conv}}$ & NS \\ 
	\hline 
	$I_{3322}$ & 0 & \blue{0.2500} & 0.2509 & \red{0.2224} & \red{0.2359} & \red{0.2487} & 0.3333 & 0.5000 & 0.6667   & 1\\ 
	$I_{3422}^1$ & 0 & \blue{0.2761} & 0.2761 & \red{0.1998} & \red{0.1998} & 0.2761 & 0.3333 & 0.5000 & 0.6667   & 1\\ 
	$I_{3422}^2$ & 0 & \blue{0.2990} & 0.2990 & \red{0.2538} & \red{0.2769}  & \red{0.2769} & 0.3333 & 0.6667 & 0.6667 &   1\\ 
	$I_{3422}^3$ & 0 & \blue{0.2910} & 0.2910 & \red{0.1893} & \red{0.2599} & \red{0.2616} & \red{0.2222} & 0.6667 & 0.6667   & 1\\ 
	$I_{3522}$ & 0 & \blue{0.3229} & 0.3229 & \red{0.2145} &\red{0.2675 } & \red{0.2675} & \red{0.2222} & 0.6667 & 0.6667  & 1\\ 
	 $I_{\text{chain3}}$ & 0 & \blue{0.5981} & 0.5981 &\red{0.0000} & \red{0.4142} & \red{0.4142} & \red{0.0000} & 1 &  1 & 1\\ 
 $I_{\text{E}}$ & 0 & \blue{0.1547} & 0.1547 &\red{0.0000} &\red{0.1381}  & \red{0.1381} & \red{0.0000} & \red{0.0000} &   0.3333 & 1\\ 
 $I_{\text{chainCHSH}}$ & 0 & \blue{0.4142} & 0.4142 &\red{0.0000}  & \red{0.2761} & \red{0.2761} & \red{0.0000} & 0.6667 &   0.6667 & 1\\ 
$M_{3322}$ & 0 & \blue{0.0122} & 0.0647 &\red{0.0000} &\red{0.0000} & \red{0.0000} & \red{0.0000} & \red{0.0000} & \red{0.0000}  & 1\\ 
	\hline \hline 
	\end{tabular} 
	\caption{Maximal value of some Bell {expressions} with respect to several constraints. The ``L'' and ``NS'' columns show the local (set to $0$) and nonsignaling (set to $1$) bounds, respectively. The column ''Qubits,'' with values in blue, reports a lower bound for the maximal violation achieved with two-qubit states (see the Appendix for details). The column ''$Q_3$'' gives the maximal value given by the third level of the NPA hierarchy \cite{NPA}, and provides an upper bound on the maximal value that can be found within QT. From column ``$L^{Q_3}_{\cap}$'' to column `` $L^{NS}_{2\text{conv}}$'', we give the bounds found by solving \eqref{e:max ineq} for different types of compatibility structures on {Alice's} side, where $NS$ or $Q_3$ indicates whether the nonsignaling constraints or the third level of the {NPA hierarchy was used}, respectively. A violation of {any of} these bounds rules out the {corresponding compatibility} structure. We have depicted {in} red the bounds that are smaller than the qubit bound, indicating that the {compatibility} structure can be ruled out in two-qubit experiments.}\label{t:results}
\end{table*}


{\em Testing incompatibility structures in the EPR-steering scenario.---}We finally consider the EPR-steering scenario, where no assumptions on Alice's measurements or the shared state are made but Bob can perform state-tomography on his part of the system \cite{wiseman07}. The experiment can be described by an assemblage $\sigma_{a|x}:=\tr_A \left( M_{a|x}\otimes \openone \; \rho \right)$, which represents the unnormalized states held by Bob when Alice performs the measurements labelled by $x$ and obtains the outcome $a$. We show that for any structure $\mathcal{C}=[{C}_1,{C}_2,\ldots,{C}_k]$, there exists a physical assemblage that {allows} us to rule out $\mathcal{C}$. This assemblage is given by local measurements $\{M_{a|x}\}$ applied on any pure entangled state with full Schmidt rank (\eg, the maximally entangled state). This extends the connection between measurement compatibility and EPR-steering established in Refs.~\cite{quintino14,uola14,kiukas17}. See {the} Appendix for more details.


{\em Conclusions and open questions.---}In this Letter, we have shown that different structures of measurement compatibility give rise to constraints in the correlations that can be observed in Bell tests. These constraints can be interpreted as a partial locality, where the {behaviors} can be nonlocal but are seen to be local when restricted to some measurement choices. As a consequence, the violation of Bell inequalities by models satisfying incompatibility structures are reduced with respect to models in which measurements can be arbitrarily incompatible. This fact allows us to test different types of measurement incompatibility in a device-independent way. 

Some open questions follow from our work. First, can any structure of genuine measurement incompatibility (for any number of measurements and outcomes) be realized by quantum system? That would generalize the results of Ref.\ ~\cite{fritz14}, where the authors have shown that any measurement structure can be realized in quantum mechanics. Also, can any structure of genuine measurement incompatibility be device-independently ruled out in QT (\ie, using quantum behaviors)? A second problem is that of mathematically characterizing the partially local sets for other scenarios and, in particular, finding tight inequalities that {limit} them. 


\acknowledgements{The authors thank Teiko Heinosaari for interesting discussions. MTQ acknowledges support from the Japan Society for the Promotion of Science (JSPS) by KAKENHI grant No.\ 16F16769. CB acknowledges support from the Austrian Science Fund (FWF): M 2107 (Meitner-Programm) and ZK 3 (Zukunftskolleg). AC acknowledges support from the Spanish MICINN Project No.\ FIS2017-89609-P with FEDER funds and the Knut and Alice Wallenberg Foundation. DC acknowledges support from the Ramon y Cajal fellowship. {DC and EW acknowledge support from the} Spanish MINECO (QIBEQI, Project No.\ FIS2016-80773-P, and Severo Ochoa SEV-2015-0522), Fundaci\'o Cellex, Generalitat de Catalunya (SGR875 and CERCA Program), and ERC CoG QITBOX. We thank the Benasque Center for Science, where this project was conceived and developed. }


\bibliographystyle{linksen}
\bibliography{DIMIbibliography}

\clearpage


{\section{Appendix A: Genuine triplewise incompatibility}}


As mentioned in the main text, we say that a set of three measurements is genuinely triplewise incompatible if it cannot be written as {a} convex combination of pairwise compatible ones. A trivial example of three measurements that are incompatible but not genuinely triplewise incompatible is given by a set of three measurements where one is the {uniformly} random POVM, with elements given by $\frac{\openone}{d}$, and the other two are incompatible. A more elaborated example is illustrated in Fig.~2 in the main text, where a set of {measurements} admits a decomposition in three sets of noisy Pauli measurements.

\begin{definition}[{Genuine triplewise incompatibility}]
A set of three measurements $\{M_{a|x}\}$
is genuinely triplewise incompatible when it cannot be written as {a} convex {combination} of measurements that are pairwise compatible on a given partition. More specifically, let $\left\{J^{12}_{a|x}\right\}$ be a set of three measurements $\left(x\in\{1,2,3\}\right)$ such that the measurements $x=1$ and $x=2$ are jointly measurable, $\left\{J^{23}_{a|x}\right\}$ a set of three measurements such that the measurements $x=2$ and $x=3$ are jointly measurable, and analogously for $\left\{J^{{13}}_{a|x}\right\}$, where measurements $x=1$ and $x=3$ are compatible. The set $\{M_{a|x}\}$ is genuinely triplewise incompatible if it \emph{cannot} be written as
\begin{equation}
 M_{a|x}=p_{12} J^{12}_{a|x} + p_{23} J^{23}_{a|x} + p_{13} J^{13}_{a|x} 
\end{equation}
for some probabilities $p_{12}$, $p_{23}$, and $p_{13}$ that respect $p_{12}+p_{23}+p_{13}=1$.
\end{definition}
 
By construction, the set of measurements that are not genuinely triplewise compatible is the convex hull of all possible pairwise compatible sets and its geometrical representation is illustrated in Fig.~\ref{geometrical}. We remark the analogy with genuine {tripartite} entanglement for mixed states, where a state is said to be genuinely tripartite entangled when it cannot be written as a convex combination of bipartite-separable ones \cite{acin01}.

 	
\begin{figure}[h!]
\includegraphics[scale=.74]{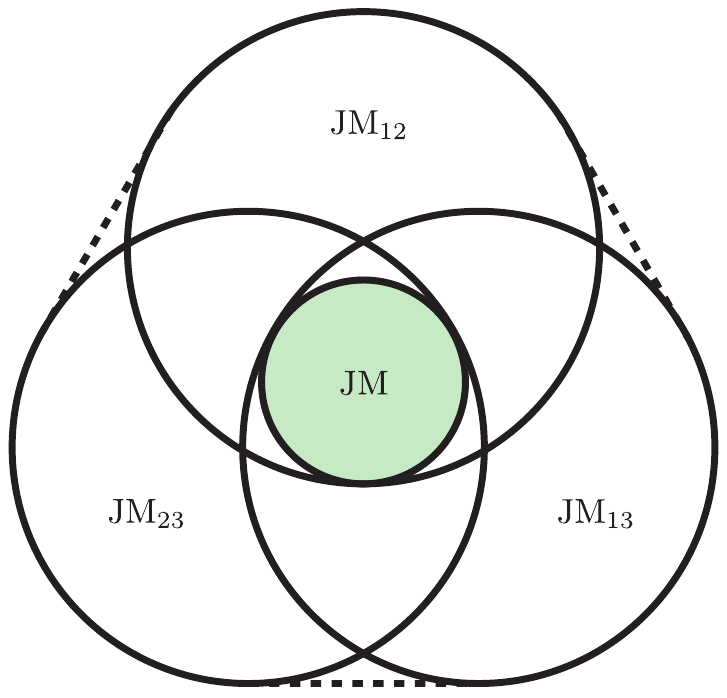} 
\caption{Geometrical interpretation of sets of three pairwise and triplewise compatible measurements. Here, the set $\text{JM}$ is the set where all three measurements are triplewise compatible. $\text{JM}_{12}$ is the set where the measurement $1$ and $2$ are compatible, and {similarly} for $\text{JM}_{23}$ and $\text{JM}_{{13}}$. Sets of measurements outside the convex hull of $\text{JM}_{12}$, $\text{JM}_{23}$, and $\text{JM}_{13}$ are the genuinely triplewise incompatible ones.}\label{geometrical}
\end{figure}

 
\section{Appendix B: More general incompatibility structures}
 

In the previous section, we have restricted ourselves to the scenario with three measurements. However, the concepts and methods used in the previous section can be {generalized} to any compatibility hypergraph. A particular case of interest is that of measurements that are genuinely $n$-wise incompatible, \ie, that cannot be written as convex combinations of $n-1$-wise compatible measurements, but even this notion can be extended to any possible incompatibility structure.

\begin{definition}[{Genuine $\mathcal{C}$-incompatibility}]
Let $\mathcal{C}=\left[C_1,C_2, \ldots, C_N\right]$ represent some compatibility structure. A set {of} measurements $\{M_{a|x}\}$ is genuinely $\mathcal{C}$-incompatible if it cannot be written as {a} convex {combination} of measurements that respect the compatibility structures $C_1$, $C_2$,\ldots, and $C_N$. More precisely, a set of measurements $\{M_{a|x}\}$ is genuinely $\mathcal{C}$-incompatible if it cannot be written as
\begin{equation}
M_{a|x}=\sum_i p_i J^{C_i}_{a|x},
\end{equation}
where $\left\{J^{C_i}_{a|x}\right\}$ are sets {of} measurements respecting the compatibility {structures} $C_i$ and $\{p_i\}$ is a probability distribution.
\end{definition}
 
In this language, the case of genuine triplewise {incompatibility} corresponds to {genuine} $\mathcal{C}$-{incompatibility} with {the choice $\mathcal{C}=[\{1,2\},\{2,3\},\{1,3\}]$}.


\section{Appendix C: SDP formulation for general compatibility}


Similarly to standard measurement compatibility (cf.\ Refs.~\cite{HeinosaariPRA15,PuseyJOB15,uola15} and related measures for EPR-steering \cite{skrzypczyk14,piani15}), the problem of deciding whether a set of measurements is genuinely incompatible for some given structure can be phrased in terms of a {semidefinite} program (SDP). We now state explicitly an SDP that decides if a set of three $d$-dimensional measurements $\{M_{a|1}\}, \{M_{a|2}\},\{M_{a|3}\}$ is genuinely triplewise incompatible. The SDP formulation for more general structures follows straightforwardly.
%

\begin{align}\label{SDP}
\text{Given}& \text{ three POVMs}\quad \{M_{a|1}\}, \; \{M_{a|2}\}, \; \{M_{a|3}\} \\
\text{find}& \quad J^{12}_{a|x}, \; J^{23}_{a|x}, \; J^{31}_{a|x}, \; p_{12},\; p_{23},\; p_{31},\; E^{12}_{\lambda}, \; E^{23}_{\lambda}, \; E^{31}_{\lambda} \nonumber \\
\text{such that}&\quad E_\lambda^{12},E_\lambda^{23},E_\lambda^{13}\geq0, \; \; p_{12}, p_{23}, p_{13}	\geq0, \nonumber\\ 
&\quad M_{a|x}=J^{12}_{a|x}+J^{23}_{a|x} + J^{13}_{a|x} ,\nonumber\\ 
&\quad J^{12}_{a|x}\geq 0, \; \forall a,x ; \; \sum_a J^{12}_{a|x} = p_{12} {\id}, \; \forall x, \nonumber\\ 
&\quad J^{12}_{a|x} = \sum_\lambda D_{\lambda}(a|x) E^{12}_\lambda \quad \text{ for } x=1,x=2 , \nonumber \\
&\quad J^{23}_{a|x}\geq 0, \; \forall a,x ; \; \sum_a J^{23}_{a|x} = p_{23} {\id}, \; \forall x, \; \nonumber\\ 
&\quad J^{23}_{a|x} = \sum_\lambda D_{\lambda}(a|x) E^{23}_\lambda \quad \text{ for } x=2,x=3 ,\nonumber \\
&\quad J^{13}_{a|x}\geq 0, \; \forall a,x ; \; \sum_a J^{13}_{a|x} = p_{13} {\id}, \; \forall x, \nonumber\\ 
&\quad J^{13}_{a|x} = \sum_\lambda D_{\lambda}(a|x) E^{13}_\lambda \quad \text{ for } x=1,x=3, \nonumber 
\end{align}
where $D_{\lambda}(a|x) $ is the set of all deterministic probability distributions in the given scenario.

One can also quantify triplewise incompatibility of a set of measurements using standard SDP methods. Here we present a semidefinite maximisation problem that quantifies how robust the triplewise incompatibility of a set $\{M_{a|x}\}$ is to white noise:
\begin{align}\label{SDP2}
\text{Given}& \{M_{a|x}\} \\
\text{{maximise} }& \eta \nonumber \\
\text{s.t. }& \eta M_{a|x} + (1-\eta) \tr(M_{a|x})\frac{\openone}{d}= J^3_{a|x} \ \forall a,x\nonumber \\
\text{where }& \{J^3_{a|x}\} \text{ feasible solution of problem ~\eqref{SDP}} \nonumber
\end{align}

The SDP problem in Eq.~\eqref{SDP2} can be expanded by inserting explicitly the problem in Eq.~\eqref{SDP} and further {simplifying} as follows:
\begin{align}\label{SDP3}
&\text{Given } \{M_{a|x}\} \\
&\text{{maximise} } \eta \nonumber \\
&\text{s.t. } E^{st}_{\lambda} \geq 0\ \forall \lambda,\ \sum_\lambda E^{st}_\lambda = \frac{\openone}{d} \sum_\lambda \tr[E^{st}_\lambda]\nonumber\\
& \text{ for } (s,t)=(1,2),(1,3),(2,3), \nonumber \\
& \frac{1}{d}\sum_\lambda \tr[E^{12}_\lambda + E^{13}_\lambda + E^{23}_\lambda ] = 1, \nonumber \\
&\eta M_{a|x} + (1-\eta) \tr(M_{a|x})\frac{\openone}{d} \geq \sum_\lambda D_{\lambda}(a|x) 
(E^{sx}_\lambda + E^{tx}_\lambda)\nonumber \\ 
&\forall a \text{ and } (s,t,x)=(1,2,3),(1,3,2),(2,3,1), \nonumber
\end{align}
where we used the convention $E^{xy}_\lambda = E^{yx}_\lambda$ to keep the notation more compact. Such a simplified version can be obtained by noticing that each $J^3_{a|x}$ is given by the sum $J^3_{a|x}= J^{sx}_{a|x}+J^{tx}_{a|x} + J^{st}_{a|x}$, where $J^{sx}_{a|x}$ and $J^{tx}_{a|x}$ arise each from a joint measurement and $J^{st}_{a|x}$ is positive, that $\{E_\lambda^{st}\}$ is proportional to a POVM, and that $\{E^{12}_\lambda + E^{13}_\lambda + E^{23}_\lambda\}$ is a POVM.

We notice that other measures of triplewise incompatibility based {on} robustness and steering weight follow directly from this SDP formulation. We refer to Ref.~\cite{dani_paul_review} for an overview {of} these measures and how to phrase them as SDPs.

Similarly, we can define a robustness {measure} with respect to arbitrary noise as 
\begin{equation}\label{eq:def_rob}
\begin{split}
t^*=\min \left\lbrace\ t \ \middle| J^3_{a|x} = \frac{M_{a|x} + t N_{a|x}}{1+t}, \text{ for } \{J^3_{a|x}\} \text{ sol.\ of} \text{ ~\eqref{SDP}, } \right.\\
\left. \{N_{a|x}\} \text{ meas. assemb. }\right\rbrace.
\end{split}
\end{equation}

For a feasible $t${, the} condition in Eq.~\eqref{eq:def_rob} can be rewritten as
\begin{equation}
\begin{split}
\sum_\lambda D_\lambda(a|x) (E^{sx}_\lambda + E^{tx}_\lambda) + J^{st}_{a|x} = \frac{M_{a|x} + t N_{a|x}}{1+t},\\
\Rightarrow (1+t)\left(\sum_\lambda D_\lambda(a|x) (E^{sx}_\lambda + E^{tx}_\lambda) + J^{st}_{a|x}\right) \geq M_{a|x} .
\end{split}
\end{equation}
By re-absorbing $1+t$ in the {normalisation} of $E^{st}_\lambda$ and $J^{st}_{a|x}$, we obtain the following SDP:
\begin{align}\label{SDP4}
\text{Given} &\{M_{a|x}\} \\
\text{{minimise} }& \frac{1}{d}\sum_\lambda \tr[E^{12}_\lambda + E^{13}_\lambda + E^{23}_\lambda ] \nonumber \\
\text{s.t. } & E^{st}_{\lambda} \geq 0\ \forall \lambda,\ \sum_\lambda E^{st}_\lambda = \frac{\openone}{d} \sum_\lambda \tr[E^{st}_\lambda]\nonumber\\
& \text{ for } (s,t)=(1,2),(1,3),(2,3), \nonumber \\
& J^{st}_{a|x} \geq 0, \ \forall a, \sum_a J^{st}_{a|x} = \sum_\lambda E^{st}_\lambda \nonumber \\
& \text{ for } (s,t,x)=(1,2,3),(1,3,2),(2,3,1), \nonumber \\
& \sum_\lambda D_\lambda(a|x) (E^{sx}_\lambda + E^{tx}_\lambda) + J^{st}_{a|x} \geq M_{a|x},\nonumber
\end{align}	
which gives as solution $1+t^*$, \ie, the robustness $+1$. Notice that this problem {clearly has} a strictly feasible solution (\ie, with strict inequality constraints satisfied), \eg, take each $E^{st}_\lambda = \openone$ and the corresponding $J^{st}_{a|x}$ coming from the linear constraints. As a consequence, Slater's condition is satisfied and the optimal values of the primal and dual {problems} coincide \cite{boyd04}.

We have implemented {code} to obtain the white noise robustness of genuine triplewise incompatibility of general $d$-dimensional measurements. Our {code} can be {found} at the online repository \cite{mtq_github_incompatibility} and can be freely used and edited.


{\section{Appendix D: General incompatibility witnesses}}


Since the set of {nongenuinely triplewise incompatible} measurements is convex, the separating hyperplane theorem states that there is always a genuine triplewise incompatibility witness {that can detect any given set of genuinely triplewise incompatible measurements} \cite{boyd04}. That is, there exists a set of operators $\{F_{a|x}\}$ acting on the same space {as} the measurements and a constant bound $\beta$ such that all nongenuinely triplewise {incompatible} measurements $\{M_{a|x}\}$ respect 
\begin{equation}\label{eq:witn_ineq}
 \sum_{a,x} \tr (F_{a|x} M_{a|x}) \leq \beta,
\end{equation}
{but the genuinely triplewise incompatible measurements under consideration violate this bound.}

For instance, such {a} witness can be obtained from any solution {to} the dual of the SDP in Eq.~\eqref{SDP4}. In fact, by substituting the equality constraints $ \sum_\lambda E^{st}_\lambda = \frac{\openone}{d} \sum_\lambda \tr[E^{st}_\lambda]$ and $\sum_a J^{st}_{a|x} = \sum_\lambda E^{st}_\lambda$ with two inequality constraints, one obtains the SDP in the standard form \cite{boyd04}
\begin{align}\label{SDP_st}
\text{Given} &\{C,B, \Phi\}\\
\text{{minimise} }& \mean{C,X} \nonumber \\
\text{s.t. } & \Phi[X] \geq B \nonumber\\
& X\geq 0 \nonumber, 
\end{align}	
which has as dual problem
\begin{align}\label{SDP_dual}
\text{Given} &\{C,B, \Phi\}\\
\text{{maximise} }& \mean{B,Y} \nonumber \\
\text{s.t. } & \Phi^\dagger[Y] \leq C \nonumber\\
& Y\geq 0 \nonumber. 
\end{align}	

By comparing Eq.~\eqref{SDP_st} with Eq.~\eqref{SDP4}, one notices that $B$ is written in terms of the given measurements $\{M_{a|x}\}$, since all other inequalities involve only the variables $E^{st}_\lambda$ and $J^{st}_{a|x}$.
As a consequence, the expression ``$\text{{maximise} } \mean{B,Y}$'' could be rewritten as ``$\text{{maximise} }\sum_{a,x} \tr (F_{a|x} M_{a|x})$'' and the value of such {an} expression, by strong duality, will correspond to the optimal value of the primal problem $1+t^*$, where $t^*$ is the robustness appearing in Eq.~\eqref{eq:def_rob}. Hence, by explicitly constructing the operators $\{F_{a|x}\}$ in terms of the matrix $B$ of the primal problem, we can certify genuine triplewise incompatibility with a violation of the condition
\begin{equation}
 \sum_{a,x} \tr (F_{a|x} M_{a|x}) \leq 1.
\end{equation}


\begin{table}
	\begin{tabular}{cccc}
	\hline \hline
	Quantum & $\text{JM}$ & $\text{JM}_2$ & $\text{JM}_3$ \\ 
	\hline 
	$6$ & $\frac{6}{\sqrt{3}} \approx 3.464 $ & $ \frac{6}{\sqrt{2}} \approx 4.242 $& $ 2(\sqrt{2}+1) \approx 4.828$ \\ 
	\hline \hline
	\end{tabular} 
	\caption{Table {summarising} the maximal attainable value that can be obtained on the compatibility witness given by $\tr (XM_1 + YM_2 + ZM_3)$, see Eq.~\eqref{witness}.
	The column ``Quantum'' gives the maximum {value} attainable by general unconstrained qubit measurements. ``$\text{JM}$'' stands for the usual joint measurability, where all three measurements are compatible. ``$\text{JM}_2$'' indicates measurements that are pairwise compatible and ``$\text{JM}_3$'' measurements that are not genuinely triplewise incompatible.
	}\label{t:results2}
\end{table}

We now present another example of {a} genuine triplewise incompatibility witness by exploring a known standard compatibility witness.
Consider three dichotomic qubit measurements described by $\{M_{a|x}\}$. We define the associated observable of a particular measurement by $M_x = M_{1|x}-M_{2|x}$ and the value $\beta$ of the witness for a particular measurement by
\begin{equation} \label{witness}
\tr (XM_1 + YM_2 + ZM_3) = \beta,
\end{equation}
where $X$, $Y$, {and} $Z$ are the Pauli {operators}.

Exploring results on steering {witnesses} \cite{cavalcanti09,skrzypczyk14} and {their} strong connection with joint measurability \cite{quintino14,uola14}, one can show that fully compatible measurements can obtain at most $\beta_{\text{JM}} = \frac{6}{\sqrt{3}}$ and pairwise {compatible measurements} (in all possible pairs)  $\beta_{\text{JM}_2} = \frac{6}{\sqrt{2}}$. {Using concepts from SDP, we can also show that nongenuinely triplewise incompatible measurements can attain at most $\beta_{\text{JM}_3} = 2(\sqrt{2}+1)$. Since the witness \eqref{witness} is linear and the set of nongenuinely triplewise incompatible measurements is nothing but the convex hull of the sets of pairwise compatible measurements $\text{JM}_{12}$, $\text{JM}_{13}$, and $\text{JM}_{23}$, we need only determine the maximal value of \eqref{witness} in each of these three sets and take the maximum.}

{We do this here for $\text{JM}_{12}$; the maximal values for $\text{JM}_{13}$ and $\text{JM}_{23}$ will inevitably be the same due to the symmetry of the witness. In this case, we want to show that
\begin{equation}
  \label{witness_JM12}
  \tr(X M_{1} + Y M_{2} + Z M_{3}) \leq 2 \sqrt{2} + 2
\end{equation}
whenever the measurements underlying $M_{1}$ and $M_{2}$ are compatible. Since there is no constraint involving $M_{3}$, clearly $\tr(Z M_{3}) \leq 2$, and we only need to prove that $\tr(X M_{1} + Y M_{2}) \leq 2 \sqrt{2}$. Substituting explicitly an underling POVM,
\begin{align}
  M_{1} = E_{1} + E_{2} - E_{3} - E_{4} , \\
  M_{2} = E_{1} - E_{2} + E_{3} - E_{4} ,
\end{align}
the term in the witness we want to maximise can be written as
\begin{equation}
  \label{witness_povms}
  \tr\bigro{(X + Y) (E_{1} - E_{4}) + (X - Y) (E_{2} - E_{3})}
\end{equation}
with the conditions $E_{\lambda} \geq 0$ and $\sum_{\lambda} E_{\lambda} = \id$. The dual of this maximisation problem, whose solution is an upper bound on \eqref{witness_povms}, can be written compactly as
\begin{IEEEeqnarray}{u?rCl}
  minimise & \IEEEeqnarraymulticol{3}{l}{\tr(\sigma)} \\
  subject to & \sigma &\geq& X + Y , \nonumber \\
  & \sigma &\geq& X - Y , \nonumber \\
  & \sigma &\geq& -X + Y , \nonumber \\
  & \sigma &\geq& -X - Y . \nonumber
\end{IEEEeqnarray}
This problem has as a feasible solution $\sigma^{*} = \sqrt{2} \id$, for which $\tr(\sigma^{*}) = 2 \sqrt{2}$, proving that \eqref{witness_povms} is upper bounded by $2 \sqrt{2}$. Finally, to see that \eqref{witness_JM12} is tight, note that the upper bound $2 \sqrt{2} + 2$ is attained with
\begin{IEEEeqnarray}{c+t+c}
  M_{1} = M_{2} = (X + Y)/\sqrt{2} &and& M_{3} = Z , \IEEEeqnarraynumspace
\end{IEEEeqnarray}
where $M_{1}$ and $M_{2}$ are compatible.}

{Hence, any set of measurements that attains
\begin{equation}
  \tr (XM_1 + YM_2 + ZM_3) > 2(\sqrt{2} + 1)
\end{equation}
is genuinely triplewise incompatible.}


\section{Appendix E: Numerical methods for the device-independent case}


Deciding whether a set of probabilities given by $\{p(ab|xy)\}$ is Bell local can be phrased in terms of linear programming (LP) \cite{NL_review}. With similar ideas, we can also write {an} LP {to test} whether a {probability distribution} can arise from a model with partial compatibility given by a compatibility structure $\mathcal{C}=[{C}_1,{C}_2,\ldots,{C}_k]$,
\begin{align}\label{LP}
\text{Given}&\quad \{P(ab|xy) \}_{abxy}, \mathcal{C}\\
\text{find}& \quad \{ p_{\lambda_i} \}_{\lambda_i} \nonumber \\
\text{s.t.}&\quad P(ab|xy)=\sum_{\lambda_i} p_{\lambda_i} D_{\lambda_i}(a|x)D_{\lambda_i}(b|y)\nonumber \\
&~ \forall x\in {C}_i, ~\forall i,~\forall a,b,y \nonumber, \\ 
&\quad p_{\lambda_i}\geq0~ \sum_{\lambda_i}p_{\lambda_i}=1\quad \forall \lambda_i, i, \nonumber
\end{align}
where {$\lambda_i$} is the local hidden variable associated to the compatibility subset ${C}_i$. 

We can also have {an} LP {characterisation} for the set $L_\mathcal{C}^{NS}$, where the probabilities are local\footnote{Strictly speaking, nongenuinely $\mathcal{C}$-incompatible. } in the compatibility structure $\mathcal{C}$ and all nonsignaling constraints are {respected}. For that, we just need to notice that the nonsignaling constraints
\begin{subequations}
\begin{align}
\sum_a p(ab|x' y)=\sum_a p(ab|x'' y) ~\forall~x',x'' ,\\
\sum_a p(ab|xy')=\sum_a p(ab|xy'') ~\forall~y',y''
\end{align}
\end{subequations}
are linear{; hence,} we can just add the nonsignaling constraints to the ones of the LP of Eq.~\eqref{LP}.

For device-independent certification of genuine $\mathcal{C}$-incompatibility, we define the set $L_\mathcal{C}^{Q}$, which imposes the constraints of Eq.~\eqref{LP} and that the full distribution $\{p(ab|xy)\}$ admits a quantum {realisation}. Deciding if a set of probabilities $\{p(ab|xy)\}$ admits a quantum realisation is known to be a very hard problem but an outer approximation of the set of distributions with {a} quantum {realisation} can be made via the NPA hierarchy \cite{NPA}. The NPA hierarchy consists {of} a set of outer approximations that {converge} to the set of distributions with {a} quantum {realisation}. Each step of this hierarchy admits {an} SDP {characterisation;} hence, by adding this SDP constraint to the LP of Eq.~\eqref{LP}, we can certify genuine $\mathcal{C}$-incompatibility in quantum mechanics. 

For the maximal qubit violation, we obtain lower bounds by explicitly providing the state and measurements. We make use of see-saw method that exploits the semidefinite program presented {in Appendix~C} of Ref.~\cite{quintino14}. Let $\gamma_{ab|xy}$ be the coefficients of a Bell {expression} that is written as
\begin{equation}
B=\sum_{abxy} \gamma_{ab|xy} p(ab|xy).
\end{equation}
Any qudit quantum probability can be written as $p(ab|xy)= \tr\left( A_{a|x} \sigma_{b|y} \right)$, where $\{A_{a|x}\}$ is a valid $d$-dimensional POVM and $ \sigma_{b|y}=\tr \left( \openone\otimes B_{b|y}\; \rho\right)$ is an assemblage {defined} by a set of POVMs $\{B_{b|y}\}$ and a bipartite quantum state $\rho$. In order to obtain a lower bound for the maximal qudit violation we choose a random\footnote{For our calculations, we have sampled our measurements over the uniform Haar measure on projective measurements. However, the method works for any sampling measure.} set of measurements for Alice and use an SDP provided in Ref.~\cite{quintino14} to obtain the assemblage $\{\sigma_{b|y}\}$ that attains the maximal quantum violation for {the} fixed measurements $\{A_{a|x}\}$. We now fix the optimal assemblage $\{\sigma_{b|y}\}$ obtained in the previous step to perform another SDP, now {optimising} over all possible choices of measurements for Alice. Iterations of this method provide a lower bound for optimal qudit violation.
	
All {code} used to construct Table~\Rnum{1} of the main text can be found in the online repository at \cite{mtq_github_incompatibility} and can be freely used and edited.


\section{Appendix F: Device-independent tests of quantum genuine triplewise incompatibility}
\label{di_quantum_jm}

In the main text we defined the set $L_{12}^Q$ which consists of behaviors $\{p(ab|xy)\}$ that are Bell local when Alice's measurements are restricted to $x=1$ and $x=2$ and can be generated by quantum local measurements. Verifying that a quantum behavior is outside of $L_{12}^Q$ is one way to certify that Alice's measurements associated to $x=1$ and $x=2$ are incompatible. As mentioned in the main text, however, since locality does not necessarily imply measurement compatibility, the restrictions of $L_{12}^Q$ are not necessarily equivalent to enforcing that Alice's quantum measurements $A_1$ and $A_2$ are compatible.

To treat this problem formally we define the set $Q_{12_\text{JM}}$, which consists of behaviors that admit a quantum realisation in terms of local measurements on a quantum state and where Alice's quantum measurements associated to $x=1$ and $x=2$ are compatible. More precisely, $\{p(ab|xy)\} \in Q_{12_\text{JM}}$ if there exists a quantum state $\rho$ and sets of POVMs $\{A_{a|x}\}$ and $\{B_{b|y}\}$ where the POVMs $\{A_{a|1}\}$, $\{A_{a|2}\}$ are compatible and $p(ab|xy)=\tr (\rho_AB A_{a|x}\otimes B_{b|y}), \forall a,b,x,y$.

In order to characterize $Q_{12_\text{JM}}$, we note that two POVMs $\{A_{a|1}\}$ and $\{A_{a|2}\}$ are compatible if and only if there exists a common underlying POVM $\{M_{a_1,a_2}\}$ such that $A_{a_1|1}=\sum_{a_2} A_{a_1,a_2} \forall a_1$ and $A_{a_2|2}=\sum_{a_1} A_{a_1,a_2} \forall a_2$ \cite{heinosaari16}. We can then characterize the set $Q_{12_\text{JM}}$ by enforcing that there is a single measurement $\{M_{a_1,a_2}\}$ that is associated to $x=1$ and $x=2$. Since we are working in a device-independent scenario, we can also always take this measurement to be projective without loss of generality, which is equivalent to saying that the operators $A_{a_{1}|1}$ and $A_{a_{2}|2}$ are projective and commute with each other. We can thus, alternatively and equivalently, characterize $Q_{12_\text{JM}}$ by enforcing that the operators of Alice's measurements $\{A_{a|1}\}$ and $\{A_{a|2}\}$ commute. Either constraint can, for instance, be used in the NPA hierarchy to obtain outer approximations for the set $Q_{12_\text{JM}}$, in the latter case as additional equality constraints between the measurement operators \cite{navascues12b}.

Analogously to $Q_{12_\text{JM}}$, we can also define and similarly characterize the sets $Q_{23_\text{JM}}$ and $Q_{13_\text{JM}}$, as well as their convex hull, $Q_{2\text{conv}_\text{JM}}$. It follows by construction that if a quantum behavior is outside $Q_{2\text{conv}_\text{JM}}$, Alice's measurements are necessarily genuinely triplewise incompatible. Since compatible measurements can only lead to local statistics, we have that $L_{2\text{conv}}^Q \supseteq Q_{2\text{conv}_\text{JM}}$.

A natural question then arises: are $Q_{2\text{conv}_\text{JM}}$ and $L_{2\text{conv}}^Q$ the same set? To answer this, we compare the values of the elegant Bell expression of Ref.~\cite{gisin91} (reproduced as Eq.~\eqref{I_E} below) that are attainable with behaviors in $Q_{2\text{conv}_\text{JM}}$ and $L_{2\text{conv}}^Q$ and show that they are different.

In Section~\ref{sos_decompositions} below we prove that, for behaviors in $Q_{2\text{conv}_{\text{JM}}}$, the elegant Bell expression $I_{\text{E}}$ respects the tight upper bound
\begin{equation}
  \label{elegant_Q2convJM_bound}
  I_{\text{E}} \leq 2 + 2 \sqrt{5} .
\end{equation}
This translates to about $0.07869$ in the scaling used in Table~\Rnum{1} in the main text, significantly less than the corresponding bound of $0.1381$ that we found for $L_{2\text{conv}}^Q$. However, since the quantum bound for $L_{2\text{conv}}^Q$ was obtained numerically at level~3 of the NPA hierarchy, which we only know to be a relaxation of the quantum set, the gap between $0.07869$ and $0.1381$ strictly speaking only establishes that $Q_{2\text{conv}_\text{JM}} \neq L_{2\text{conv}}^{Q_{3}}$. To prove that $Q_{2\text{conv}_\text{JM}} \neq L_{2\text{conv}}^Q$, however, we need only exhibit a behavior in $L_{2\text{conv}}^Q$ that violates the $Q_{2\text{conv}_\text{JM}}$ bound \eqref{elegant_Q2convJM_bound}.

Just such a behavior can be constructed by adjusting the ideal quantum strategy that maximimally violates the elegant Bell inequality. More precisely, the behavior is obtained by Alice and Bob performing measurements of the form
\begin{IEEEeqnarray}{c+c+c}
  A_{1} = X, & A_{2} = Y, & A_{3} = Z
\end{IEEEeqnarray}
and
\begin{align}
  B_{1} &= \frac{\sin(\mu)}{\sqrt{2}} \bigro{X - Y}
  + \cos(\mu) Z , \\
  B_{2} &= \frac{\sin(\mu)}{\sqrt{2}} \bigro{X + Y}
  - \cos(\mu) Z , \\
  B_{3} &= \frac{\sin(\mu)}{\sqrt{2}} \bigro{-X - Y}
  - \cos(\mu) Z , \\
  B_{4} &= \frac{\sin(\mu)}{\sqrt{2}} \bigro{-X + Y}
  + \cos(\mu) Z
\end{align}
on the maximally-entangled state $\ket{\phi^{+}} = (\ket{00} + \ket{11})/\sqrt{2}$. With this strategy the elegant Bell expression attains
\begin{equation}
  \label{elegant_mu}
  I_{\text{E}} = 4 \sqrt{2} \sin(\mu) + 4 \cos(\mu) ,
\end{equation}
depending on the parameter $\mu$. We now impose that the behavior is contained in $L_{12}^Q$, which we recall is a subset of $L_{2\text{conv}}^Q$. To this end, we require that the behavior satisfies all of the CHSH inequalities involving $A_{1}$ and $A_{2}$. One can verify that the highest of the relevant CHSH expectation values is
\begin{equation}
  S_{\text{max}_{12}} = 2 \sqrt{2} \abs{\sin(\mu)} ;
\end{equation}
thus, the behavior is in $L_{12}^{Q}$ provided that $\abs{\sin(\mu)} \leq 1/\sqrt{2}$. Making the specific choice $\sin(\mu) = \cos(\mu) = 1/\sqrt{2}$, then, according to \eqref{elegant_mu} we attain a value of the elegant Bell expression,
\begin{equation}
  I_{\text{E}} = 4 \Bigro{1 + \frac{1}{\sqrt{2}}},
\end{equation}
that translates to approximately $0.1381$ in the scaling of Table~\Rnum{1} and significantly violates the $Q_{2\text{conv}_\text{JM}}$ bound \eqref{elegant_Q2convJM_bound} given above. In this way we confirm that the set $Q_{2\text{conv}_\text{JM}}$ is strictly smaller than $L_{2\text{conv}}^Q$.



\section{Appendix G: Facets of the $L_{2\text{\lowercase{conv}}}^{NS}$ polytope}
\label{L2convNS_facets}

As we mentioned in the main text, the sets $L_{2\text{conv}}^{NS}$ are polytopes. As such, they can be completely characterized either as the convex hulls of finite numbers of vertices or as the sets of behaviors satisfying finite numbers of inequalities corresponding to their facets.

One can in general derive a polytope's facets given its vertices (or vice versa) using existing algorithms, although in practice the problem scales badly and rapidly becomes intractable for large polytopes. In the simplest case, however, where Alice and Bob each have three inputs and two outputs, we were able to exactly characterize $L_{2\text{conv}}^{NS} = \conv\bigro{L_{12}^{NS} \cup L_{13}^{NS} \cup L_{23}^{NS}}$ and derive all of its facets.

We briefly describe the procedure we followed. The sets $L_{12}^{NS}$, $L_{13}^{NS}$, and $L_{23}^{NS}$ are all fully characterized by a finite number of inequalities that we already know, corresponding to the positivity and some of the CHSH facets of the local polytope. We first used the software PORTA \cite{PORTA} (which can solve vertex and facet enumeration problems using exact rational arithmetic) to derive the vertices of these three sets. We then took the union of the three sets of vertices (of which $L_{2\text{conv}}^{NS}$ is the convex hull) and used PORTA again to find the corresponding facets. Finally, we grouped together facets that were equivalent to each other up to relabellings of inputs and outputs (but not parties).

The polytope $L_{2\text{conv}}^{NS}$ turns out to have a total of $4452$ facets which can be grouped into six equivalence classes that we list here. In terms of the full- and single-body correlators, which we define as
\begin{align}
  \mean{A_xB_y} &= p(a=b|xy) - p(a\neq b|xy) , \\
  \avg{A_x} &= p(a=0|x) - p(a=1|x) , \\
  \avg{B_{y}} &= p(b=0|y) - p(b=1|y) ,
\end{align}
$L_{2\text{conv}}^{NS}$ has $36$ inequalities representing positivity ($p(ab|xy) \geq 0$) conditions which are relabelling equivalent to
\begin{equation}
  \label{F1}
  F_{1} = \mean{A_1} + \mean{B_1} -\mean{A_1B_1} \leq 1,
\end{equation}
$384$ inequalities which are equivalent to
\begin{align} \label{F2}
F_{2}=& \phantom{+} \mean{A_1}  +\mean{A_2}+\mean{A_3}\\
&+  \mean{B_1}  +\mean{B_2}+\mean{B_3} \nonumber \\ &
+\mean{A_1B_1}-\mean{A_1B_2}-\mean{A_1B_3} \nonumber \\ &
-\mean{A_2B_1}+\mean{A_2B_2}-\mean{A_2B_3} \nonumber \\ &
-\mean{A_3B_1}-\mean{A_3B_2}+\mean{A_3B_3} \leq 7 \nonumber,
\end{align}
$576$ inequalities which are equivalent to
\begin{align} 
F_{3}=& \phantom{+} \mean{A_1}  +\mean{A_2}\\
&+\mean{A_1B_1}+\mean{A_1B_2}+\mean{A_1B_3} \nonumber \\ &
+\mean{A_2B_1}-\mean{A_2B_2}-\mean{A_2B_3} \nonumber \\ &
+2\mean{A_3B_2}-2\mean{A_3B_3} \leq 8 \nonumber,
\end{align}
$1152$ inequalities which are equivalent to
\begin{align} 
F_{4}=& \phantom{+} \mean{A_1}  +\mean{A_2}+\mean{A_3}\\
&+  \mean{B_1}  +\mean{B_2}+\mean{B_3} \nonumber \\ &
+2\mean{A_1B_1}+2\mean{A_1B_2}-\mean{A_1B_3} \nonumber \\ &
+\mean{A_2B_1}-2\mean{A_2B_2}-2\mean{A_2B_3} \nonumber \\ &
-2\mean{A_3B_1}+\mean{A_3B_2}-2\mean{A_3B_3} \leq 11 \nonumber,
\end{align}
$1152$ inequalities which are equivalent to
\begin{align}
  \label{F5}
  F_{5}=& \phantom{+}  \mean{A_1}+\mean{A_2}+\mean{A_3}  + \mean{B_1}\\
  &+2\mean{A_1B_1}+2\mean{A_1B_2}-\mean{A_1B_3} \nonumber \\ &
  +2\mean{A_2B_1}-\mean{A_2B_2}-2\mean{A_2B_3} \nonumber \\ &
  +\mean{A_3B_1} - 3\mean{A_3B_2} + 3\mean{A_3B_3} \leq 13 \nonumber,
\end{align}
and $1152$ inequalities which are equivalent to
\begin{align} \label{F6}
F_{6}=& \phantom{+} \mean{A_1}  +\mean{A_2}+\mean{B_1}+\mean{B_2}\\
&+\mean{A_1B_1}-\mean{A_1B_2}+\mean{A_1B_3} \nonumber \\ &
+\mean{A_2B_1}-\mean{A_2B_2}-\mean{A_2B_3} \nonumber \\ &
+\mean{A_3B_1}+\mean{A_3B_2} \leq 6 . \nonumber
\end{align}
Of these, the inequality \eqref{F6} is equivalent to the $M_{3322}$ inequality of Ref.~\cite{BGS15} (this can be checked, for instance, with the algorithms of faacets \cite{faacets}) and is the only one violated in quantum physics. For the other five Bell expressions, the local and quantum bounds are identical to the $L_{2\text{conv}}^{NS}$ bounds given in \eqref{F1}--\eqref{F5}. They are thus examples of Bell inequalities with no quantum violation. This is obvious for the positivity inequality \eqref{F1}. The proofs that \eqref{F2}--\eqref{F5} represent the quantum bounds of the four remaining Bell expressions, $F_{2}$ to $F_{5}$, are given in Section~\ref{sos_decompositions}.
	
We have also completely characterized the correlation polytope for this scenario, that is, the projection of $L_{2\text{conv}}^{NS}$ involving only the full-body correlators $\mean{A_x B_y}$. The correlation polytope has $18$ facets representing positivity which are relabelling equivalent to 
\begin{equation}
  \label{FC1}
  FC_{1} = \mean{A_xB_y} \leq 1 ,
\end{equation}
$192$ facets which are relabelling equivalent to 
\begin{align}  \label{FC2}
FC_{2}=& \phantom{+}  2\mean{A_1B_1}+2\mean{A_1B_2}+\mean{A_1B_3}  \\ &
+2\mean{A_2B_1}-\mean{A_2B_2}-2\mean{A_2B_3} \nonumber \\ &
+\mean{A_3B_1}-2\mean{A_3B_2}+2\mean{A_3B_3} \leq 11 \nonumber,
\end{align}
and $288$ facets which are relabelling equivalent to 
\begin{align} \label{FC3}
FC_{3}=& \phantom{+}  \mean{A_1B_1}+\mean{A_1B_2}+\mean{A_1B_3}  \\ &
+\mean{A_2B_1}+\mean{A_2B_2}-\mean{A_2B_3} \nonumber \\ &
+\mean{A_3B_1}-\mean{A_3B_2} \leq 6 \nonumber.
\end{align}
Neither \eqref{FC2} or \eqref{FC3} can be violated or even attained in quantum physics. $FC_{3}$ has a local bound of $4$ and a quantum bound of $5$, while the local and quantum bounds of $FC_{2}$ are both $9$. Both quantum bounds are proved in Section~\ref{sos_decompositions}.



\section{Appendix H: Proofs of quantum bounds}
\label{sos_decompositions}

In the two previous sections we asserted exact quantum bounds for several Bell expressions. We collect proofs of these together in this section. All of the proofs below are given in the form of a sum-of-squares (SOS) decomposition for the corresponding quantum Bell operator together with an explicit quantum realisation showing that the bound can be attained in quantum physics.

SOS decompositions are closely related \cite[Section~21.2.6]{navascues12b} to the NPA hierarchy that we used to obtain the numeric quantum bounds reported in the main text. In fact, an SOS decomposition is simply a way of giving a feasible solution to the dual SDP of a given level of the NPA hierarchy in a factored form that is positive semidefinite by construction. In the context of Bell nonlocality, an SOS decomposition was first used in \cite{tsirelson80} to prove the quantum bound $S \leq 2 \sqrt{2}$ of the CHSH expression.

We were able to derive some of the simpler SOS decompositions below by trial and error. The more complicated ones were derived following a semi-systematic method that is described in Section~3.4 of \cite{woodhead18} (see also \cite{bamps15}).

\subsection{Correlation polytope facets}

We begin with $FC_{3}$ to illustrate the technique. Proving that the bound $FC_{3} \leq 5$ asserted in Section~\ref{L2convNS_facets} holds in quantum physics is equivalent to proving that the shifted Bell operator $5 \id - B_{FC_{3}}$ is positive semidefinite, where
\begin{align}
B_{FC_{3}} =& \phantom{+} A_1 B_1 + A_1 B_2 + A_1 B_3  \\ &
+ A_2 B_1 + A_2 B_2 - A_2 B_3 \nonumber \\ &
+ A_3 B_1 - A_3 B_2 \nonumber
\end{align}
is the Bell operator associated to $FC_{3}$, $A_x$ and $B_y$ are related to Alice's and Bob's measurement operators by $A_x = A_{0|x} - A_{1|x}$ and $B_y = B_{0|y} - B_{1|y}$, and spacelike separation of the parties implies that $A_x$ and $B_y$ commute, $\forall x,y$. We recall that, since we are making no assumption restricting the dimension of the Hilbert space, we can also and, in the following, will take the measurements to be projective without loss of generality, in which case ${A_x}^{2} = {B_y}^{2} = \id$.

Following these observations, we can prove that $5 \id - B_{F_{3}}$ is positive semidefinite by expressing it in the form of a sum of squares,
\begin{align} 
  5 \id - B_{FC_{3}} = {} & \frac{1}{2}(A_1 + A_2 - B_1- B_2)^2 \\
  &+ \frac{1}{2}(A_1 - A_2 - B_3)^2 + \frac{1}{2}(B_1 - B_2 - A_3)^2 .
  \nonumber
\end{align}
The right side is manifestly positive semidefinite, since it is a sum of squares of self-adjoint operators, and expanding and simplifying it by substituting the projectivity and commutation rules $A_{x}^{2} = B_{y}^{2} = \id$ and $B_{y} A_{x} = A_{x} B_{y}$ yields exactly the shifted Bell operator on the left. Hence, for quantum behaviors, $FC_{3}$ is  bounded by $5$. To see that the bound is tight, we note that it is attained by measuring
\begin{IEEEeqnarray}{rCcCl}
  A_{1} &=& B_{1} &=& \frac{\sqrt{3}}{2} Z + \frac{1}{2} X \,, \\
  A_{2} &=& B_{2} &=& \frac{\sqrt{3}}{2} Z - \frac{1}{2} X \,, \\
  A_{3} &=& B_{3} &=& X
\end{IEEEeqnarray}
on the maximally-entangled two-qubit state $\ket{\phi^{+}} = (\ket{00} + \ket{11})/\sqrt{2}$.

The quantum bound $FC_2 \leq 9$ is also proved by a similarly simple SOS decomposition,
\begin{align} 
9 \id - B_{FC_{2}} = {} & \frac{1}{6}\bigro{3A_1 -(2B_1+2B_2+B_3)}^2 \\
&+ \frac{1}{6}\bigro{3A_2 -(2B_1-B_2-2B_3)}^2 \nonumber \\
&+ \frac{1}{6}\bigro{3A_3 -(B_1-2B_2+2B_3)}^2 . \nonumber
\end{align}
The maximum value $FC_{2} = 9$ is attained with the local deterministic strategy $A_{1} = A_{2} = B_{1} = B_{2} = +1$ and $A_{3} = B_{3} = -1$.

\subsection{$L_{2\text{conv}}^{NS}$ facets}

In the previous section we claimed that the quantum bounds of the facet expressions $F_{2}$, $F_{3}$, $F_{4}$, and $F_{5}$ are identical to their local and $L_{2\text{conv}}^{NS}$ bounds. We prove this here by giving an SOS decomposition for each of the shifted Bell operators as well as a local deterministic strategy that attains the bound.

For $F_2$, the quantum bound $F_{2} \leq 7$ is proved by the SOS
\begin{align}
  7\id - B_{F_2} = \frac{1}{4} \Bigro{ & \babs{R_1^{++}}^2
    + \babs{R_1^{-+} - R_2^{-+}}^2 \\
  & + 2\babs{R_2^{-+}}^2 + 2\babs{R_1^{--}}^2 } , \nonumber
\end{align}
where $\abs{R}^{2} = R^{\dagger} R$ and
\begin{align}
  R^{++}_{1} &= 2 \id - A_{1} - A_{2} - A_{3} - B_{1} - B_{2} - B_{3} , \\
  R^{-+}_{1} &= A_{1} + A_{2} - B_{1} - B_{2} , \nonumber \\
  R^{-+}_{2} &= A_{3} - B_{3} , \nonumber \\
  R^{--}_{1} &= A_{1} - A_{2} - B_{1} + B_{2} . \nonumber
\end{align}
The bound is attained with the local deterministic strategy
$A_1=A_2 = B_1=B_2=1$ and $A_3=B_3=-1$.

For $F_3$, the quantum bound $F_{3} \leq 8$ is implied by
\begin{align}
8\id - B_{F_3} = \frac{1}{4} \Bigro{ & \babs{R^{++}_{1}}^{2} + \babs{R^{++}_{2}}^{2} \\
&+ 2 \babs{R^{-+}_{1} - R^{-+}_{2}}^{2} + 2 \babs{R^{--}_{1}}^{2} } , \nonumber
\end{align}
where 
\begin{align}
  R^{++}_{1} &= 2 \id - A_{1} - A_{2} , \\
  R^{++}_{2} &= A_{1} + A_{2} - 2 B_{1} , \nonumber \\
  R^{-+}_{1} &= A_{1} - A_{2} ,\nonumber \\
  R^{-+}_{2} &= B_{2} + B_{3} , \nonumber \\
  R^{--}_{1} &= 2 A_{3} - B_{2} + B_{3} . \nonumber
\end{align}	
The local bound $8$ can be attained by setting $A_{1} = A_{2} = A_{3} = B_{1} = B_{2} = +1$ and $B_{3} = -1$.

For $F_4$ we have that
\begin{align}
  11 \id - B_{F_4} = {}
  &\frac{1}{32} \babs{2 R^{+}_{2} - R^{+}_{7} - R^{+}_{8}}^{2} \\
  &+ \frac{1}{4} \babs{R^{+}_{1} - R^{+}_{6}}^{2} + \frac{1}{2} \babs{R^{+}_{3}}^{2}   + \frac{1}{2} \babs{R^{+}_{4}}^{2} \nonumber \\
  &+ \frac{1}{32} \babs{    2 R^{-}_{1} - 2 R^{-}_{2} + R^{-}_{4} - R^{-}_{5}}^{2} \nonumber \\
  &+ \frac{1}{2}\babs{R^{-}_{1} - R^{-}_{2}}^{2} , \nonumber
\end{align}	 
where (note that not all $R_i^\pm$ we define are used)
\begin{align}
  R^{+}_{1} &= \id- A_{1} , \\
  R^{+}_{2} &= 2 \id - B_{1} - B_{2} , \nonumber \\
  R^{+}_{3} &= 2 \id - A_{1} (B_{1} + B_{2}) , \nonumber \\
  R^{+}_{4} &= 2 \id + (A_{2} + A_{3}) B_{3} , \nonumber \\
  R^{+}_{5} &= A_{2} + A_{3} - 2 B_{3} , \nonumber \\
  R^{+}_{6} &= (\id - A_{1}) B_{3} , \nonumber \\
  R^{+}_{7} &= (A_{2} + A_{3}) (2 \id - B_{1} - B_{2}) , \nonumber \\
  R^{+}_{8} &= (A_{2} - A_{3}) (B_{1} - B_{2}) , \nonumber \\
  R^{-}_{1} &= A_{2} - A_{3} , \nonumber \\
  R^{-}_{2} &= B_{1} - B_{2} , \nonumber \\
  R^{-}_{3} &= A_{1} (B_{1} - B_{2}) , \nonumber \\
  R^{-}_{4} &= (A_{2} + A_{3}) (B_{1} - B_{2}) , \nonumber \\
  R^{-}_{5} &= (A_{2} - A_{3}) (B_{1} + B_{2}) , \nonumber \\
  R^{-}_{6} &= (A_{2} - A_{3}) B_{3} . \nonumber
\end{align}
The bound $11$ is attainable by setting $  A_{1} = B_{1} = B_{2} = B_{3} = +1 $ and $ A_{2} = A_{3} = -1$.

Finally, for $F_5$ we have that
\begin{IEEEeqnarray}{rCl}
  \IEEEeqnarraymulticol{3}{l}{13 \id - B_{F_5}} \\
  \quad &=& \frac{1}{32} \babs{
    2 R^{+}_{1} + 2 R^{+}_{4} - 2 R^{+}_{5} + R^{+}_{6} - R^{+}_{8}}^{2} \nonumber \\ 
  &&{}+ \frac{1}{4} \babs{R^{+}_{2} + R^{+}_{7}}^{2}
  + \frac{1}{2} \babs{R^{+}_{1} + R^{+}_{3}}^{2}
  + \frac{1}{2} \babs{R^{+}_{5}}^{2} \nonumber \\
  &&{}+ \frac{1}{32} \babs{2 R^{-}_{1} + R^{-}_{3} - 2 R^{-}_{4} + R^{-}_{6}}^{2} \nonumber \\
  &&{}+ \frac{1}{2} \babs{R^{-}_{1} - R^{-}_{2}}^{2} , \nonumber
\end{IEEEeqnarray}	 
where 
\begin{align}
  R^{+}_{1} &= 2 \id - A_{1} - A_{2} , \\ \nonumber
  R^{+}_{2} &= \id - B_{1} , \\ \nonumber
  R^{+}_{3} &= (A_{1} + A_{2}) (\id - B_{1}) , \\ \nonumber
  R^{+}_{4} &= 2 \id + A_{3} (B_{2} - B_{3}) , \\ \nonumber
  R^{+}_{5} &= 2 A_{3} + B_{2} - B_{3} , \\ \nonumber
  R^{+}_{6} &= (2 \id - A_{1} - A_{2}) (B_{2} - B_{3}) , \\ \nonumber
  R^{+}_{7} &= A_{3} (\id - B_{1}) , \\ \nonumber
  R^{+}_{8} &= (A_{1} - A_{2}) (B_{2} + B_{3}) , \\ \nonumber
  R^{-}_{1} &= A_{1} - A_{2} , \\ \nonumber
  R^{-}_{2} &= B_{2} + B_{3} , \\ \nonumber
  R^{-}_{3} &= (A_{1} + A_{2}) (B_{2} + B_{3}) , \\ \nonumber
  R^{-}_{4} &= A_{3} (B_{2} + B_{3}) , \\ \nonumber
  R^{-}_{5} &= (A_{1} - A_{2}) B_{1} , \\ \nonumber
  R^{-}_{6} &= (A_{1} - A_{2}) (B_{2} - B_{3}) \,.
\end{align}
The local bound $13$ is attainable by setting $A_{1} = A_{2} = A_{3} = B_{1} = B_{3} = +1$ and $B_{2} = -1$.

\subsection{$Q_{2\text{conv}_\text{JM}}$ bound of the elegant Bell expression}

Finally, we prove the $Q_{2\text{conv}_\text{JM}}$ bound $I_{\text{E}} \leq 2 + 2 \sqrt{5}$ on the elegant Bell expression asserted in Section~\ref{di_quantum_jm}. In this case it is sufficient to establish that the bound holds for $Q_{12_\text{JM}}$ since symmetries of the elegant Bell expression with respect to relabellings of inputs and outputs imply that the same bound must hold for $Q_{13_\text{JM}}$ and $Q_{23_\text{JM}}$ and, therefore, also their convex hull $Q_{2\text{conv}_\text{JM}}$. An SOS decomposition proving the bound is
\begin{IEEEeqnarray}{rCl}
  \IEEEeqnarraymulticol{3}{l}{
    \label{elegant_sos}
    \bigro{2 + 2\sqrt{5}} \id - B_{I_{\text{E}}}} \\
  \qquad &=& \frac{\sqrt{5} - 1}{32}
  \babs{R^{+{+}+}_{1} + R^{+{+}+}_{2} + R^{+{+}+}_{3}}^{2} \nonumber \\
  &&{}+ \frac{3 \sqrt{5} - 5}{64} \babs{R^{+{+}+}_{1} + R^{+{+}+}_{2}}^{2}
  \nonumber \\
  &&{}+ \frac{5 - \sqrt{5}}{32} \Bigro{
    \babs{R^{+{+}-}_{1}}^{2} + \babs{R^{+{-}+}_{1}}^{2}
    + \babs{R^{-{-}+}_{1}}^{2}} \IEEEeqnarraynumspace \nonumber \\
  &&{}+ \frac{5 (\sqrt{5} - 1)}{512} 
    \begin{IEEEeqnarraybox}[][t]{rl}
      \Bigro{& \babs{R^{+{-}-}_{1} + R^{+{-}-}_{2}}^{2} \\
      &+\: \babs{R^{-{+}-}_{1} + R^{-{+}-}_{2}}^{2}}
    \end{IEEEeqnarraybox} \nonumber \\
  &&{}+ \frac{5}{512} \babs{R^{-{-}+}_{1} + R^{-{-}+}_{2}}^{2} , \nonumber
\end{IEEEeqnarray}
where
\begin{IEEEeqnarray}{rCl}
  R^{+{+}+}_{1} &=& \frac{4}{\sqrt{5}} \id
  - A_{1} \bigro{B_{1} + B_{2} - B_{3} - B_{4}} , \\
  R^{+{+}+}_{2} &=& \frac{4}{\sqrt{5}} \id
  - A_{2} \bigro{B_{1} - B_{2} + B_{3} - B_{4}} , \nonumber \\
  R^{+{+}+}_{3} &=& \Bigro{2 + \frac{2}{\sqrt{5}}} \id
  - A_{3} \bigro{B_{1} - B_{2} - B_{3} + B_{4}} , \nonumber \\
  R^{+{+}-}_{1} &=& \Bigro{2 + \frac{2}{\sqrt{5}}} A_{3}
  - \bigro{B_{1} - B_{2} - B_{3} + B_{4}} , \nonumber \\
  R^{+{-}+}_{1} &=& \frac{4}{\sqrt{5}} A_{2}
  - \bigro{B_{1} - B_{2} + B_{3} - B_{4}} , \nonumber \\
  R^{+{-}-}_{1} &=& \frac{4}{\sqrt{5}}
  A_{1} \bigro{B_{1} + B_{2} + B_{3} + B_{4}} \nonumber \\
  &&{}- \Bigro{2 - \frac{2}{\sqrt{5}}}
  A_{3} \bigro{B_{1} - B_{2} + B_{3} - B_{4}} , \nonumber \\
  R^{+{-}-}_{2} &=& \frac{4}{\sqrt{5}}
  A_{2} \bigro{B_{1} - B_{2} - B_{3} + B_{4}} \nonumber \\
  &&{}+ \Bigro{2 + \frac{2}{\sqrt{5}}}
  A_{3} \bigro{B_{1} - B_{2} + B_{3} - B_{4}} , \nonumber \\
  R^{-{+}+}_{1} &=& \frac{4}{\sqrt{5}} A_{1}
  - \bigro{B_{1} + B_{2} - B_{3} - B_{4}} , \nonumber \\
  R^{-{+}-}_{1} &=& \frac{4}{\sqrt{5}}
  A_{1} \bigro{B_{1} - B_{2} - B_{3} + B_{4}} \nonumber \\
  &&{}+ \Bigro{2 + \frac{2}{\sqrt{5}}}
  A_{3} \bigro{B_{1} + B_{2} - B_{3} - B_{4}} , \nonumber \\
  R^{-{+}-}_{2} &=& \frac{4}{\sqrt{5}}
  A_{2} \bigro{B_{1} + B_{2} + B_{3} + B_{4}} \nonumber \\
  &&{}- \Bigro{2 - \frac{2}{\sqrt{5}}}
  A_{3} \bigro{B_{1} + B_{2} - B_{3} - B_{4}} , \nonumber \\
  R^{-{-}+}_{1} &=& \Bigro{2 - \frac{2}{\sqrt{5}}}
  A_{1} \bigro{B_{1} - B_{2} + B_{3} - B_{4}} \nonumber \\
  &&{}+ \frac{4}{\sqrt{5}}
  A_{3} \bigro{B_{1} + B_{2} + B_{3} + B_{4}} , \nonumber \\
  R^{-{-}+}_{2} &=& \Bigro{2 - \frac{2}{\sqrt{5}}}
  A_{2} \bigro{B_{1} + B_{2} - B_{3} - B_{4}} \nonumber \\
  &&{}+ \frac{4}{\sqrt{5}}
  A_{3} \bigro{B_{1} + B_{2} + B_{3} + B_{4}} . \nonumber
\end{IEEEeqnarray}
Under the standard projectivity and commutation rules, the right side of \eqref{elegant_sos} expands and simplifies to
\begin{align}
  &\bigro{2 + 2\sqrt{5}} \id - B_{I_{\text{E}}} \\
  &\qquad - \frac{3 \sqrt{5} - 5}{32} \comm{A_{1}}{A_{2}}
  \bcomm{B_{1} - B_{4}}{B_{2} - B_{3}} , \nonumber
\end{align}
which in turn simplifies to $\bigro{2 + 2\sqrt{5}} \id - B_{I_{\text{E}}}$ if $\comm{A_{1}}{A_{2}} = 0$. The bound is attained, for instance, by preparing the maximally-entangled state $\ket{\phi^{+}} = (\ket{00} + \ket{11})/\sqrt{2}$ and measuring
\begin{IEEEeqnarray}{c+c}
  A_{1} = A_{2} = Z , & A_{3} = X
\end{IEEEeqnarray}
on Alice's side and
\begin{IEEEeqnarray}{rCl+rCl}
  B_{1} &=& \frac{2}{\sqrt{5}} Z + \frac{1}{\sqrt{5}} X , &
  B_{2} &=& - X , \\
  B_{4} &=& - \frac{2}{\sqrt{5}} Z + \frac{1}{\sqrt{5}} X , &
  B_{3} &=& - X
\end{IEEEeqnarray}
on Bob's side.



\section{Appendix I: Other Bell inequalities considered}


Here we list all Bell inequalities used in the main text. For the scenario where both Alice and Bob have three measurements, the $I_{3322}$ inequality is the only {facet inequality} that is inequivalent to CHSH up to relabelling\footnote{Notice that, because CHSH inequalities are precisely those {characterising} the $L_{ij}^{\rm NS}$ sets, $I_{3322}$ is the only relevant inequality for the case of three dichotomic measurements for Alice and for Bob. Similarly, for more complex scenarios discussed later on, the CHSH inequality will also be irrelevant.}. The $I_{3322}$ inequality first appeared in {\cite{froissart81} and was later independently derived in \cite{sliwa03,collins04}. It} can be written as 
\begin{align} \label{I3322}
I_{3322}=& -\mean{A_1} - \mean{A_2}+\mean{B_1}+\mean{B_2} \\
&+ \mean{A_1B_1} + \mean{A_1B_2}+\mean{A_1B_3} \nonumber \\ &
+\mean{A_2B_1}+\mean{A_2B_2}-\mean{A_2B_3} \nonumber \\ &
 +\mean{A_3B_1}-\mean{A_3B_2} \leq 4 \nonumber.
\end{align}

The scenario where Alice can perform three measurements and Bob four has three classes of tight Bell inequalities that are not relabelling-equivalent to the CHSH and the $I_{3322}$ inequality. This scenario was first {characterized} in Ref.~\cite{collins04} and the three classes of inequalities are {represented} by
\begin{align} \label{I3422_1}
M_{3422}^1= &p_A(0|1)+p_A(0|2)-2p_A(0|3) \\ 
 + &p_B(0|1)+p_B(0|4)\nonumber \\ 
 -&p(00|11) -p(00|12)+p(00|13)-p(00|14) \nonumber \\ 
 -&p(00|21) +p(00|22)-p(00|23)-p(00|24) \nonumber \\ 
 +&p(00|31) +p(00|32)+p(00|33)-p(00|34)\leq 2 ,\nonumber 
\end{align}
\begin{align} \label{I3422_2}
M_{3422}^2= &p_A(0|2)-p_A(0|3) \\ 
 - &p_B(0|1)-p_B(0|3) +p_B(0|4) \nonumber \\ 
 -&p(00|11) +p(00|13)-p(00|14) \nonumber \\ 
 +&p(00|21) -p(00|22)-p(00|24) \nonumber \\ 
 +&p(00|31) +p(00|32)+p(00|33)\leq 1 ,\nonumber 
\end{align}
\begin{align} \label{I3422_3}
M_{3422}^3= &p_A(0|1)-p_A(0|3) \\ 
 - &p_B(0|3)+2p_B(0|4)\nonumber \\ 
 -&2 p(00|11) +p(00|13)-p(00|14) \nonumber \\ 
 -&p(00|21) -p(00|22)+p(00|23)-p(00|24) \nonumber \\ 
 +&p(00|31) +p(00|32)+p(00|33)-p(00|34)\leq 2. \nonumber 
\end{align}

The scenario where Alice can perform three measurements and Bob five has only one class of tight Bell inequalities that are not relabelling-equivalent to the ones in previous scenarios. This scenario was first {characterized} in Ref.~\cite{quintino14} and this new inequivalent inequality is given by
\begin{align} \label{I3522}
I_{3522}= &\phantom{+} \mean{B_1} + \mean{B_2} \\
 &+\mean{A_1B_1} - \mean{A_1B_2}+\mean{A_1B_3}+\mean{A_1B_4}
\\ &+\mean{A_2B_1} - \mean{A_2B_2} - \mean{A_2B_4} +\mean{A_2B_5} \nonumber \\ 
	&+ \mean{A_3B_1}-\mean{A_3B_2} -\mean{A_3B_3}-\mean{A_3B_5} \leq 6.
 \nonumber
\end{align}

The chained Bell inequality we used \cite{Pearle70,braunstein90} is given by
\begin{align} \label{chain3}
I_{\text{chain3}}&= \mean{A_1B_1} + \mean{A_2B_1} \\ 
&+ \mean{A_2B_2} + \mean{A_3B_2} \nonumber \\
&+ \mean{A_3B_3} - \mean{A_1B_3} \leq 6. \nonumber
\end{align}

The elegant Bell inequality we used was introduced in Ref.~\cite{Elegant} and can be defined as
\begin{align} \label{I_E}
I_{\text{E}}= &+ \mean{A_1B_1} + \mean{A_1B_2} - \mean{A_1B_3} - \mean{A_1B_4} \\ 
 &+ \mean{A_2B_1} - \mean{A_2B_2} + \mean{A_2B_3}- \mean{A_2B_4} \nonumber \\
 &+ \mean{A_3B_1} - \mean{A_3B_2} - \mean{A_3B_3}+ \mean{A_3B_4} \leq 6. \nonumber
\end{align}

The chained version of {the} CHSH inequality we used was proposed in Ref.~\cite{APVW16} and reads as
\begin{align} \label{chain_CHSH}
I_{\text{chainCHSH}}&= \mean{A_1B_1} + \mean{A_1B_2}+ \mean{A_2B_1}- \mean{A_2B_2} \\ 
&+\mean{A_1B_3} + \mean{A_1B_3}+ \mean{A_3B_4}- \mean{A_3B_4} \nonumber \\
&+\mean{A_2B_5} + \mean{A_2B_5}+ \mean{A_3B_6}- \mean{A_3B_6} \leq 6. \nonumber
\end{align}

The $M_{3322}$ Bell inequality we used was presented in Ref.~\cite{BGS15}. {It} is relabelling equivalent to our inequality $F_6$, a facet of the $L_{2\text{conv}}^{NS}$ polytope{,} and is given by
\begin{align} \label{M3322}
M_{3322}= &-2p_A(0|1)- p_A(0|2)-2p_B(0|1)
\\ &+ p(00|11) +p(00|12)+p(00|13) \nonumber \\ 
 &+p(00|21) +p(00|22)-p(00|23) \nonumber \\ 
 &+p(00|31) - p(00|32)\leq 0.\nonumber
\end{align}


\section{Appendix J: Every compatibility structure can be ruled out in the EPR-steering scenario}


In this section we consider the problem of certifying measurement incompatibility {in the semi-device-independent EPR-steering} scenario. Let $\rho\in L(\mathbb{C}^d\otimes \mathbb{C}^d)$ be a quantum state and $\{M_{a|x}\}$ be a set of quantum measurements.
An assemblage given by $\sigma_{a|x}=\tr_A \left(M_{a|x}\otimes \openone \; \rho \right)$ is unsteerable if it can be written as
\begin{equation}
 \sigma_{a|x}=\sum_\lambda {\pi_{\lambda}} p(a|x,\lambda) \rho_\lambda,
\end{equation}
where $\pi_{\lambda}$ is a probability distribution on $\lambda$, $p(a|x,\lambda)$ is a distribution on $a$, and {$\rho_{\lambda}$} are quantum states. References~\cite{quintino14,uola14} show that a set of measurements is compatible (jointly measurable) if and only if it is useful for EPR-steering. More precisely, let $\ket{\psi}$ be a full {Schmidt}-rank entangled state and $\{M_{a|x}\}$ a set of quantum measurements. The set $\{M_{a|x}\}$ is compatible if and only if the assemblage $\sigma_{a|x}=\tr_A \left(M_{a|x}\otimes \openone \; \ketbra{\psi}{\psi} \right)$ is steerable. In this section we extend this previous result to any general compatibility structure. 
	
Given {a set of compatibility structures} $\mathcal{C}=\left[C_1,C_2, \ldots, C_N\right]$, we say that the assemblage $\{\sigma_{a|x}\}$ is genuinely $\mathcal{C}$-steerable when it cannot be written 
 \begin{equation}
 \sigma_{a|x}=\sum_i p_i \tau^{C_i}_{a|x},
 \end{equation}
 where $\left\{\tau^{C_i}_{a|x}\right\}$ is an unsteerable assemblage when structure $C_i$ is considered and $\{p_i\}$ is a probability distribution.
From the above definition we see that if {an} assemblage given by $\sigma_{a|x}$	is genuinely $\mathcal{C}$-steerable, one certifies that the measurements $\{M_{a|x}\} $ held by Alice are genuinely $\mathcal{C}$-incompatible in a semi-device-independent way. We can now present the main theorem of this section.

\begin{theorem}
Let $\{M_{a|x}\}$, $M_{a|x}\in L\left(\mathbb{C}^d\right)$ be a set of measurements and $\mathcal{C}=\left[C_1,C_2, \ldots, C_N\right]$ be some compatibility structure. If $\{ M_{a|x}\}$ is not genuinely $\mathcal{C}$-incompatible, the assemblage $\sigma_{a|x}=\tr_A \left(M_{a|x}\otimes \openone \; \rho \right)$ is not genuinely $\mathcal{C}$-steerable, regardless {of} the state $\rho$. 

If $\{ M_{a|x}\}$ is genuinely $\mathcal{C}$-incompatible, the assemblage ${\sigma_{a|x}=\tr_A \left(M_{a|x}\otimes \openone \; \ketbra{\psi}{\psi} \right)}$ is genuinely $\mathcal{C}$-steerable, where $\ket{\psi}$ {is} a $d$-Schmidt-rank pure entangled state.
\end{theorem}

\begin{proof}
If $\{M_{a|x}\}$ is $\mathcal{C}$-compatible, it can be written as
\begin{equation}
 M_{a|x}=\sum_i p_i J^{C_i}_{a|x},
 \end{equation}
where $\left\{J^{C_i}_{a|x}\right\}$ are sets {of} measurements respecting the compatibility structure $C_i$ and $\{p_i\}$ is a probability distribution. Since for every fixed structure $C_i$, the sets of measurements $\{J^{C_i}_{a|x}\}$ are compatible, they can be decomposed as
\begin{equation}
J^{C_i}_{a|x}=\sum_\lambda p^{C_i}(a|x,\lambda) E_\lambda^{C_i},
\end{equation} 
where
\begin{subequations}
\begin{align}
p^{C_i}(a|x,\lambda) & \geq 0, \; \forall a,x,\lambda,i, \\ 
\sum_a p^{C_i}(a|x,\lambda) & =1, \; \forall x, \lambda, i,
\end{align}
\end{subequations}
and $E_\lambda^{C_i}$ are valid POVMs on $\lambda$. 
We can then construct a local hidden state model for the assemblage $\sigma_{a|x}=\tr_A \left(M_{a|x}\otimes \openone \; \rho \right)$ respecting the structure $\mathcal{C}$ by setting $\rho_\lambda^{C_i}: = \frac{ \tr_A \left( E_\lambda^{C_i}\otimes I \rho \right) }{ \tr \left( E_\lambda^{C_i}\otimes I \rho \right) }$, $	p_\lambda^{C_i} = \tr \left( E_\lambda^{C_i}\otimes I \rho \right)$, and using the same distributions $p_i $ and $p^{C_i}(a|x,\lambda)$.
	
We now show that if $\{ M_{a|x}\}$ is genuinely $\mathcal{C}$-incompatible, then the assemblage ${\sigma_{a|x}=\tr_A \left(M_{a|x}\otimes \openone \; \ketbra{\psi}{\psi} \right)}$ is genuinely $\mathcal{C}$-steerable.
Let $\ket{\phi^+_d}=\sum_{i=0}^{d-1} \ket{ii}$ be the {unnormalized} $d$-dimensional maximally-entangled state. Any $d$-Schmidt-rank state $\ket{\psi}\in \mathbb{C}^d\otimes \mathbb{C}^d$ can be written as ${\ket{\psi} = \left( D\otimes \openone\right) \ket{\phi_d^+}}$, where $D$ is {an} invertible positive {semidefinite operator} that is diagonal in the $\{\ket{i}\}_{i=0}^{d-1}$ basis {and }respects $\tr D^2=1$. A straightforward calculation shows that the assemblage generated by the measurements $\{M_{a|x}\}$ on $\ket{\psi}$ is given by
\begin{equation} \label{eq:relation}
\sigma_{a|x}=	\tr_A \left(M_{a|x}\otimes \openone \; \ketbra{\psi}{\psi} \right)= D M_{a|x}^T D ,
\end{equation}
where $T$ stands for the transposition in the basis $\{\ket{i}\}_{i=0}^{d-1}$. 

Assume, by contradiction, that the assemblage ${\sigma_{a|x}=\tr_A \left(M_{a|x}\otimes \openone \; \ketbra{\psi}{\psi} \right)}$ is not genuinely $\mathcal{C}$-steerable. We can then decompose the assemblage as 
\begin{equation} \label{eq:relation2}
 \sigma_{a|x}=\sum_i p_i \tau^{C_i}_{a|x},
 \end{equation}
 where
\begin{equation}
\tau^{C_i}_{a|x}=\sum_\lambda \pi_\lambda^{C_i} p^{C_i}(a|x,\lambda) \rho_\lambda^{C_i},
\end{equation} 
 with $\pi_\lambda^{C_i}$ being a probability distribution on $\lambda$ and $\rho_\lambda^{C_i}$ valid quantum states.
 We can then construct a $\mathcal{C}$-compatible model for the measurements $\{M_{a|x}\}$ by setting ${E^{C_i}_\lambda= D^{-1} \left(\rho_\lambda^{C_i}\right)^T D^{-1}}$.
It follows from Eqs.~\eqref{eq:relation} and \eqref{eq:relation2} that $\sum_\lambda \rho_\lambda^{C_i}=(D D)^T$ for all $i$, and since $D$ is diagonal, $(D D)^T=D D$.
We now observe that {the} operators $E_\lambda^{C_i}$ form a valid quantum measurement in $\lambda$ since
\begin{align}
\sum_\lambda E_\lambda^{C_i}= \sum_\lambda D^{-1} \left(\rho_\lambda^{C_i}\right)^T D^{-1} = \openone .
\end{align} 
We now see that our $\mathcal{C}$-compatible measurement model respects
\begin{align}
M_{a|x}=\sum_i p_i \left( \sum_\lambda p^{C_i}(a|x,\lambda) E_\lambda^{C_i} \right),
\end{align}
thus contradicting the hypothesis that $\{M_{a|x}\}$ is $\mathcal{C}$-compatible.
		
\end{proof}


\end{document}